\documentclass{article}
\usepackage{graphicx} % Required for inserting images
\usepackage{geometry}
\usepackage{setspace}
\usepackage{booktabs}
\usepackage{array}
\newcolumntype{L}[1]{>{\raggedright\arraybackslash}p{#1}}
\usepackage{float}
\usepackage{amsmath} 
\usepackage{amsthm}
\usepackage{amssymb}
\usepackage{bbm} % mathbbm for indicator function
\usepackage{setspace} % for double spacing
\usepackage{bm}
\usepackage{natbib}
\usepackage[backref]{hyperref} 
\usepackage[title]{appendix}
\usepackage{enumitem}

\allowdisplaybreaks

\newtheorem{proposition}{Proposition}
\newtheorem{condition}{Condition}
\newtheorem{theorem}{Theorem}
\newtheorem{remark}{Remark}
\newtheorem{lemma}{Lemma}

\newcommand{\normm}[1]{\lVert#1\rVert}

\newcommand{\E}{\mathbb{E}}
\newcommand{\R}{\mathbb{R}}

\newcommand{\I}{\mathbbm{1}}
\DeclareMathOperator{\Var}{Var}
\DeclareMathOperator{\Cov}{Cov}
\DeclareMathOperator{\adj}{adj}

\title{Local Polynomial Estimation in Irregular Correlated Random Coefficient Panel Data Models}
\author{
Yelong Chen\thanks{Department of Economics, UCL. We would like to thank Ben Deaner, Liyang Sun and seminar participants at UCL Econometrics brownbag for helpful comments.}
\and Hua Jin\footnotemark[1]
\and Yuxuan Ren\footnotemark[1]
}

\date{This draft: \today}

\begin{document}

\maketitle
\normalsize
\doublespacing

\begin{abstract}
We study estimation of the average partial effect (APE) in an irregular correlated random coefficients panel data model in which the time dimension equals the number of random coefficients, $T=p$, and the density function is degenerate. We extend \cite{grahampowell} in two directions. First, we allow the density of $d_i=\det(X_i)$ to vanish or diverge at the origin at the polynomial rate $|d|^{a}$ with $a>-1$, and show that the aggregate time effects and the APE remain identified for every such
$a$, extending earlier work studying the case $a=0$.
Second, we propose a three-step estimator: an $m$th-order local polynomial regression over the stayers for the time effects; 
a local polynomial fit of the conditional mean coefficient just outside the trimming threshold; and an imputation step that extrapolates the latter into the stayer region rather than discarding stayers. We establish joint asymptotic normality in three regimes determined by the sign of $a$, and obtain a convergence rate approaching $n^{-(m+1)/(2m+3)}$ when $a=0$, against $n^{-1/3}$ for \cite{grahampowell}, whose estimator is the special case $m=0$, $a=0$. The results show that smoothness of the underlying conditional expectations can be traded for a faster convergence rate in irregularly identified models. Simulations confirm lower bias and root mean squared error across designs. We redo the overidentification analysis of a markup model in \cite{Raval2023}, allowing for correlated random elasticities.
\end{abstract}

\newpage
\section{Introduction}
This paper studies identification and estimation of the average partial effect (APE) in a correlated random coefficients panel data model in which the time dimension equals the number of coefficients. Specifically, the panel is short, with $T$ periods, and the model is
\begin{equation}
Y_{it}=X_{it}'(\beta_i+\delta_t)+\varepsilon_{it},
\label{eq:intro_model}
\end{equation}
where $X_{it}$ is a $p\times1$ regressor vector, $\beta_i$ is a unit-specific $p\times1$ coefficient vector, and $\delta_t$ is a $p\times1$ common drift across time. The individual coefficients are allowed to be correlated with the regressor history $X_i=(X_{i1}',\ldots,X_{iT}')'$. This correlation is well understood in the literature to introduce heterogeneity bias into the usual fixed-effects estimator of the population mean $\beta_0=\mathbb{E}[\beta_i]$, the APE.

The identification framework when $T>p$ is classical. \cite{CHAMBERLAIN1982,CHAMBERLAIN1987,CHAMBERLAIN1992} establish a random-coefficient panel framework for identifying the APE and the common time drift that relies on conditional moment restrictions and sufficient variation in the regressor history. The key requirement is that $T>p$, so that the regressor matrix is not square and residual variation remains after the regressors are projected out. Then the information bound is nonsingular, and under regularity conditions a $\sqrt{n}$-type estimator is feasible.

In many applications, however, individuals are observed for only a few periods. Studies of household consumption, as in \cite{grahampowell,Sasaki_2026}, or firm production require observing the same households or firms in every period of a balanced panel. Requiring more periods to satisfy Chamberlain's $T>p$ condition leaves fewer units with complete histories and can substantially reduce the usable sample. This motivates methods for identification and estimation in the irregular case $T=p$.

\cite{grahampowell} provides the foundational identification framework in this case. Their key insight is that when $X_i$ is square, the determinant $d_i=\det(X_i)$ separates the sample into two types of units: stayers, with $d_i$ near zero and little movement in the regressor path, and movers, with $d_i$ away from zero and enough variation to identify heterogeneous coefficients. Premultiplying the model by the adjugate matrix $X_i^{*}$ gives $X_i^{*}Y_i = X_i^{*}W_i\delta_0 + d_i\beta(X_i)+X_i^{*}\varepsilon_i$\footnote{This is the model in matrix form; see Section~\ref{sec:model}.}, so for stayers the heterogeneous term vanishes and the drift is identified from the local stayer population. For movers, the coefficient is recovered from $\beta(X_i)=X_i^{-1}(Y_i-W_i\delta_0)$, and the APE can be estimated by averaging these unit-specific coefficients over the mover population after trimming away near-stayers. Their estimator is therefore a trimmed mean over movers, with the trimming threshold chosen as a bandwidth around the singular boundary.

The key limitations of this approach are precisely two. First, the trimming rule discards stayers when constructing the APE. The resulting estimator target is effectively the average among movers rather than the population APE, so the leading misspecification is a first-order trimming bias in the bandwidth. The loss of the stayer mass from the estimand induces the slow $n^{-1/3}$-type rate in the baseline case. Second, the identification argument in \cite{grahampowell} requires a local density condition of the form $f_D(d)$ bounded away from zero and infinity around the origin; this is a strong restriction on the density of the determinant at the boundary. In practice, the stayer population may be sparse or the near-singular designs may be crowded. 

This paper addresses both issues. The method is a three-step local polynomial procedure. First, it estimates the drift $\delta_0$ by a local polynomial fit at the stayer boundary, using the stayers and their near-stayer neighbors. Second, it estimates the conditional mean of the heterogeneous coefficient, $g_\beta(h)=\mathbb{E}[\beta(X_i)\mid h_i=h]$, just outside the trimming threshold, using the movers closest to the boundary. Third, it extrapolates $\hat g_\beta(h_i)$ into the stayer region and imputes the resulting values for stayers instead of discarding them. The essential conceptual point is that the procedure removes the trimming bias through extrapolation and imputation rather than by shrinking the trimming bias with increasingly aggressive undersmoothing. In the estimator, the trimmed mean is replaced by an imputation estimator that restores the population average $\mathbb{E}[\beta(X_i)]$.

The contribution is therefore twofold. The first is identification: the paper relaxes the local density condition and identifies both $\delta_0$ and $\beta_0$ under a much broader class of boundary behavior. The second is asymptotic: the local polynomial order $m$ yields a richer asymptotic characterization in which the rate depends on the local degeneracy of the determinant density. This is the central distinction from the earlier benchmark.

The identification result is stated in terms of a local density condition
\begin{equation}
 f_D(d)=|d|^{a}q_D(d),\qquad a>-1,
\end{equation}
with $q_D$ smooth and nonzero at the origin. Our paper shows that both $\delta_0$ and $\beta_0$ remain identified for every admissible $a>-1$, and $a$ only has an impact on the rate of estimation. For $a>0$ the stayer density vanishes near the origin, so the stayer population is sparse; for $a=0$ the density is locally bounded away from zero and infinity; for $a<0$ the density is locally crowded near the boundary. The asymptotic theorem delivers a joint normal limit in the three regimes. With local polynomial order $m$, the rate is governed by the boundary exponent and the undersmoothing condition. In the case $a>0$ when the stayer population is thin, the first-step drift estimation is the rate-limiting component. In the case $a<0$, the move-to-stayer crowding at the boundary is the binding constraint, and the first-step and second-step estimators converge at different rates and become asymptotically independent. At the knife-edge $a=0$, the two steps balance, and as the bandwidth is pushed to the edge of the admissible range the rate approaches
\begin{equation}
 n^{-\frac{m+1}{2m+3}}.
\end{equation}

The paper also studies finite-sample performance using calibrated simulations from a real-world Engel-curve application in \cite{grahampowell} and an empirical application motivated by \cite{Raval2023}. Simulations calibrated to the calorie-demand application show that individual demand-elasticity estimates vary significantly with the determinant $d_i$, indicating that families with greater income variation tend to have lower income elasticities. In the second application, we use four consecutive years of the Chilean ENIA census to estimate a gross-output production function with a plant-specific productivity intercept and heterogeneous elasticities for capital, labour, and intermediate inputs, together with a common time drift. We allow these plant-specific coefficients to be correlated with the plants' input histories. We then redo the overidentification analysis for the markup model.\footnote{Under cost minimization, each variable input's first-order condition delivers a separate estimate of the same firm-level markup, as the ratio of that input's output elasticity to its revenue share, so a production function estimated with more than one variable input is overidentified for the markup: the markups implied by different inputs (materials and labour) should agree if the maintained production-function assumptions are correct. \cite{Raval2023} tests exactly this restriction and rejects it, but that implementation maintains a homogeneous elasticity across producers, so a rejection could in principle reflect nothing more than heterogeneity in returns to inputs rather than a genuine failure of the production-function or cost-minimization assumptions.} We test the equality of the average markups inferred from different inputs. Because this requires only agreement of population averages rather than agreement of all the individual coefficients or their distribution, it is a weaker test than \cite{Raval2023}'s.

\textbf{Literature review.} Our work builds on \cite{grahampowell}, who extend the regular correlated random-coefficient panel framework of \cite{CHAMBERLAIN1982,CHAMBERLAIN1987,CHAMBERLAIN1992} to the irregular case \(T=p\). We extend their framework in the two directions discussed above. Another close antecedent is \cite{Sasaki_2026}, who develop local-polynomial corrections for trimming bias and inference that accommodates a mass of stayers and many slow movers. Their correction uses observations within the trimming region; ours estimates the conditional mean coefficient from movers outside it and extrapolates inward. This avoids dividing by determinants close to zero, which would amplify noise.
%We discuss the differences in detail after Theorem~\ref{theorem:asymptotic}.
%\cite{Sasaki_2026} likewise use local polynomials to recover the contribution removed by trimming, but their correction is estimated in the signed determinant from the stayer and slow-mover window $|d_i|\le h_N$ and then added to the trimmed mover mean. The derivative terms are learned from nonzero slow movers in that window, while exact stayers enter through its mass and moments. We instead estimate $g_\beta(h)$ from movers just outside the threshold and extrapolate it into the stayer region before imputing each trimmed observation. 
Their theory permits exact stayers and an asymmetric continuous density function with curvature $-1<a\le 0$ and provides asymptotically valid inference for APE; ours excludes a point mass but covers density function with curvature $a$ for every $a>-1$, including the sparse case $a>0$, and gives an asymptotic theory that implies explicit rate results across three regimes.

Other papers study related heterogeneous-coefficient problems but differ in the target parameter, panel structure, or identification strategy. \cite{PESARAN199579} analyze heterogeneous dynamic panels with coefficients that are not correlated with the regressors; \cite{ArellanoB2012} target distributional features of the random coefficients rather than the APE; long-panel work such as \cite{CorrelatedPanelModelLargeT} exploits regular asymptotics and standard bias corrections; \cite{LAAGE2024105804} studies a control-function approach in a different endogeneity design; and \cite{hausman2025linear} estimate structural and causal effects in nonseparable panel models by averaging bias-corrected, ridge-regularized individual sieve estimates, a regularization strategy for a richer nonparametric model rather than a response to the exact short-panel singularity studied here. Outside the panel setting altogether, \cite{kwon2026estimatingtreatmenteffectsbounded} studies the cross-sectional model with bounded heterogeneity bias, trading it off against variance through a generalized ridge estimator for cross-sectional treatment effects. These papers study related heterogeneity problems while imposing different correlation structures, target objects, panel dimensions, or bias-variance trade-offs. The present paper is distinguished by combining correlated random coefficients, short-panel irregular identification, and a three-step local polynomial estimator designed to restore the population APE under degenerate local density.

\begin{table}[H]
\centering
\caption{Comparison with the literature on heterogeneous coefficient panel models}
\label{tab:litreview}
\small
\singlespacing
\begin{tabular}{@{}L{3.3cm}L{2.6cm}L{2.9cm}L{3.3cm}L{2.7cm}@{}}
\toprule
\textbf{Paper} & \textbf{Target} & \textbf{Panel structure} & \textbf{Estimator} & \textbf{Convergence} \\
\midrule
\cite{grahampowell} & APE & $T=p$ (irregular) & Trimmed mover estimator & approaching $n^{-1/3}$ at the knife-edge $a=0$ \\[4pt]
\cite{Sasaki_2026} & APE & $T=p$ with stayers and slow movers & Local-polynomial aggregate correction & inference normalized by the estimated influence-function variance \\[4pt]
\cite{CHAMBERLAIN1982,CHAMBERLAIN1987,CHAMBERLAIN1992} & APE / finite-dimensional $\theta_0$ & $T>p$ (regular) & Conditional-moment / feasible GLS & regular $\sqrt{n}$ typically available \\[4pt]
\cite{PESARAN199579} & Mean / long-run coefficient & Dynamic panel, large $T$ & Mean-group and pooled estimators & consistency \\[4pt]
\cite{ArellanoB2012} & Distributional features of coefficients & Fixed $T$, factor-model structure & Distributional identification & consistency \\[4pt]
\cite{CorrelatedPanelModelLargeT} & Moments of random coefficients / APE & Large $T$ & Bias-corrected GMM & regular $n^{-1/2}$-type rates \\[4pt]
\cite{LAAGE2024105804} & APE & Short panel with time-varying endogeneity & Control-function estimator & $n^{-1/2}$ \\[4pt]
\cite{hausman2025linear} & Structural / causal effects & Nonseparable model, fixed and large $T$ & Ridge-regularized individual estimates, averaged & regular $n^{-1/2}$-type rates \\[4pt]
\cite{kwon2026estimatingtreatmenteffectsbounded} & Heterogeneous treatment effects & Cross-sectional (no panel) & Generalized ridge, bias-aware inference & not point-consistent \\[4pt]
\textbf{This paper} & APE & $T=p$ (irregular) & Local polynomial imputation estimator & approaching $n^{-(m+1)/(2m+3)}$ at $a=0$ \\
\bottomrule
\end{tabular}
\end{table}

\textbf{Organization.} The remainder of the paper is organized as follows. Section~\ref{sec:model} sets out the correlated random coefficients model, reviews the identification and estimation strategy of Graham and Powell, and extends identification to a degenerate stayer density. Section~\ref{sec:estimation} develops the three-step local polynomial estimator. Section~\ref{sec:theory} derives its asymptotic properties across the three regimes indexed by the boundary exponent $a$, together with primitive smoothness conditions. Section~\ref{sec:simulations} evaluates finite-sample performance against the Graham-Powell benchmark. Section~\ref{sec:chile} applies the method to the Chilean manufacturing production data and uses the resulting average input elasticities to redo the overidentification test for markup estimation of \cite{Raval2023} under correlated random coefficients.

\textbf{Notation.} Let $\adj(X)$ denote the adjugate matrix of a square matrix $X$ and $\det(X)$ its determinant. The symbol $\normm{\cdot}_F$ denotes the Frobenius norm and $\normm{\cdot}$ the Euclidean norm; $\I[\cdot]$ is the indicator function. $C^{m}[a,b]$ denotes the class of functions with $m$ continuous derivatives on $[a,b]$. We use $a_n\lesssim b_n$ to denote domination up to a constant factor and $a_n\asymp b_n$ when the two are of the same order.

\section{Irregular Correlated Random Coefficient Panel Data Model}\label{sec:model}
\subsection{Setup and Irregularity}

Consider the model
\begin{equation}\label{eq:model}
    Y_i = X_i\beta(X_i)+W_i\delta_0 + \epsilon_i, \quad \E(\epsilon_i|X_i)=0,\quad var(\epsilon_i|X_i)=\sigma^2I_T
\end{equation}
where $X_i=(x_{i,1},\dots,x_{i,T})'\in\R^{T\times p}$, $\beta(X_i)\in\R^p$, $Y_i\in\R^T$,
\[
W_i = \begin{pmatrix}
    0_{p}',& & 0_p'\\
    x_{i,2}',& & 0_p'\\
    \vdots&  \ddots & \vdots\\
    0_{p}',& & x_{i,T}'\\
\end{pmatrix}\in\R^{T\times p(T-1)},
\qquad
\delta_0=(\delta_{0,2}',\dots,\delta_{0,T}')'\in\R^{p(T-1)},
\qquad \epsilon_i\in\R^T .
\]
 
Equation~(\ref{eq:model}) is a reformulation of equation~(9) in \cite{grahampowell}. The unit-specific vector $\beta(X_i)$ is allowed to covary arbitrarily with the entire regressor path; $\delta_0$ collects the aggregate time effects, which are common to all units under the common-trends assumption, with the normalisation $\delta_{0,1} = \mathbf{0}$. Throughout, we take $T=p$, so that $X_i$ is square, and we write
\[
    d_i=d(X_i)=\det(X_i),\qquad h_i=|d_i|,\qquad X_i^{*}=\adj(X_i)=d_i\,X_i^{-1},
    \qquad W_i^{*}=X_i^{*}W_i,\qquad Y_i^{*}=X_i^{*}Y_i .
\]
The parameters of interest are the aggregate time effects $\delta_0$ and the average partial effect
\begin{align}\label{eq:estimand}
    \beta_0 = \E[\beta(X_i)].
\end{align}

The model is irregular when $T=p$. The natural route to $\delta_0$ is to define $M(X_i)=I_T-X_i(X_i'X_i)^{-1}X_i'$ and exploit $\E[M(X_i)Y_i\mid X_i]=M(X_i)W_i\delta_0$. This is the basis of \cite{CHAMBERLAIN1992}, and it requires residual variation in $W_i$ after projecting out $X_i$. When $T=p$ and $X_i$ is nonsingular, however, $M(X_i)=0$. Chamberlain's information bound for $\delta_0$ is then $\mathcal{I}(\delta_0)=\E[W_i'M(X_i)\Sigma(X_i)^{-1}M(X_i)W_i]=0$, and the information bound for $\beta_0$ is singular as well. By \cite{CHAMBERLAIN1992}, no regular $\sqrt{n}$-consistent estimator of either parameter exists. 

\subsection{Graham and Powell (2012) Identification and Estimation}
\paragraph{Identification of $\delta_0$ and $\beta_0$.}
Identification, as shown in \cite{grahampowell}, rests on the behaviour of the model within a vanishingly small subpopulation. This subpopulation is characterized by $d_i$. Since $d_i=0$ exactly when the rows (equivalently, columns) are linearly dependent, the magnitude of $d_i$ reflects the amount of linearly independent within-individual regressor variation. A simple case arises when regressors remain constant over time. In the leading example $T=p=2$, with $x_{it}=(1,z_{it})'$, we have $d_i=z_{i2}-z_{i1}$, so $d_i=0$ corresponds exactly to no within-individual variation in $z_{it}$. Following \cite{grahampowell}, we refer to units with $d_i$ at or near zero as \textit{stayers} and the remainder as \textit{movers}. The two subpopulations carry complementary identifying content: stayers pin down $\delta_0$ because their outcomes move only with calendar time, while movers pin down $\beta_0$ because only they exhibit the regressor variation needed to trace out a slope.

Premultiplying (\ref{eq:model}) by $X_i^{*}$ and using $X_i^{*}X_i=d_iI_p$ gives
\begin{equation}\label{eq:adj_model}
    Y_i^{*} = d_i\,\beta(X_i) + W_i^{*}\delta_0 + X_i^{*}\epsilon_i,
    \qquad\text{so}\qquad
    \E\!\left[Y_i^{*}\mid X_i\right] = W_i^{*}\delta_0 + d_i\,\beta(X_i).
\end{equation}
The heterogeneous term enters multiplied by $d_i$ and therefore disappears on the stayer subpopulation: $\E[Y_i^{*}\mid X_i, d_i=0]=W_i^{*}\delta_0$. Provided that $d_i$ is continuously distributed with a density that is positive and smooth in a neighbourhood of the origin --- so that conditioning on $d_i=0$ is the limit of conditioning on events of positive probability --- and that $\E[W_i^{*\prime}W_i^{*}\mid d_i=0]$ is nonsingular, the conditional linear predictor of $Y_i^{*}$ on $W_i^{*}$ among stayers identifies $\delta_0$,
\begin{align}\label{eq:estimand2}
    \delta_0 = \E[W_i^{*\prime}W_i^{*}|d_i=0]^{-1}\E[W_i^{*\prime}Y_i^{*}|d_i=0].
\end{align}
 
With $\delta_0$ treated as known, (\ref{eq:model}) gives $\beta(X_i)=\E[X_i^{-1}(Y_i-W_i\delta_0)\mid X_i]$ for every $X_i$ with $d_i\neq0$, which suggests identifying $\beta_0$ by the unconditional mean of $X_i^{-1}(Y_i-W_i\delta_0)$. This fails: $X_i^{-1}=d_i^{-1}X_i^{*}$, and when $d_i$ has positive density at the origin the expectation is generally undefined. \cite{grahampowell} resolve this by trimming. Under smoothness of $\E[\beta(X_i)\mid d_i=u]$ near $u=0$ and integrability of $\beta(X_i)$,
\begin{equation}\label{eq:trimmed_id}
    \beta_0 = \E[\beta(X_i)]
    = \lim_{h\downarrow0}\E\!\left[\beta(X_i)\,\I[h_i> h]\right]
    = \lim_{h\downarrow0}\E\!\left[X_i^{-1}(Y_i-W_i\delta_0)\,\I[h_i> h]\right].
\end{equation}
Trimming removes the units with near-singular design matrices; taking the limit and exploiting continuity restores the original estimand. The remaining conditions in \cite{grahampowell} are regularity conditions of the usual kind: bounded second moments of $W_i^{*}$, twice continuous differentiability in $u$ of the density of $d_i$, and of the conditional moments, $\E[\beta(X_i)\mid d_i=u]$, $\E[W_i^{*}\mid d_i=u]$ and $\E[W_i^{*\prime}W_i^{*}\mid d_i=u]$ in a neighbourhood of zero, and nonsingularity of $\E[W_i^{*\prime}W_i^{*}\mid h_i=0]$.

\paragraph{Estimation}
The estimator replaces (\ref{eq:estimand}) and (\ref{eq:trimmed_id}) by sample analogues, using a single bandwidth $h_0\to0$ to split the sample. The aggregate time effects are estimated over the stayer subsample,
\begin{equation}\label{eq:gp_delta}
    \hat{\delta}_0^{\,GP}
    = \left[\frac{1}{n}\sum_{i=1}^n \I[h_i\le h_0]\,W_i^{*\prime}W_i^{*}\right]^{-1}
      \left[\frac{1}{n}\sum_{i=1}^n \I[h_i\le h_0]\,W_i^{*\prime}Y_i^{*}\right],
\end{equation}
and the average partial effect by the trimmed mean over the mover subsample,
\begin{equation}\label{eq:gp_beta}
    \hat{\beta}_0^{\,GP}
    = \frac{\frac{1}{n}\sum_{i=1}^n \I[h_i> h_0]\,X_i^{-1}\!\left(Y_i-W_i\hat{\delta}_0^{\,GP}\right)}
           {\frac{1}{n}\sum_{i=1}^n \I[h_i> h_0]} .
\end{equation}
Estimator (\ref{eq:gp_delta}) is a local constant kernel regression of $Y_i^{*}$ on $W_i^{*}$ evaluated at the point $h=0$. Both estimators converge at rate $\sqrt{nh_0}$, and $h_0$ must satisfy $nh_0^3\to 0$ to remove the first-order trimming bias, so that the fastest attainable rate is slower than $n^{-1/3}$. Our approach strengthens the estimation by exploiting the smoothness of these conditional expectations, and thereby improves the rate of convergence.

\subsection{Identification under a Degenerate Stayer Density}

The identification argument of \cite{grahampowell} conditions on the exact stayers, $d_i=0$, which relies on the density of $d_i$ being bounded away from both zero and infinity at $d=0$. If instead the stayer density vanishes or diverges as $d\to0$, the conditional expectations are no longer well-defined at $d=0$. We therefore replace this conditioning argument with an explicit boundary limit and impose the following condition.

\begin{condition}\label{cond:identification}\leavevmode
    \begin{enumerate}
        \item\label{id:limits} The two-sided limits
        $\mathcal{A}(0^{+}):=\lim_{u\to 0}\E[W_i^{*\prime}W_i^{*}\mid d_i=u]$ and
        $\mathcal{B}(0^{+}):=\lim_{u\to 0}\E[W_i^{*\prime}Y_i^{*}\mid d_i=u]$
        exist, and $\mathcal{A}(0^{+})$ is nonsingular.
        \item\label{id:nomass} $P(d_i=0)=0$.
        \item\label{id:bounded} For some $\bar\delta>0$,
        $\displaystyle\sup_{0<|u|<\bar\delta}\normm{\E[\beta(X_i)\mid d_i=u]}
        +\sup_{0<|u|<\bar\delta}\normm{\E[W_i^{*\prime}\beta(X_i)\mid d_i=u]}\le C.$
    \end{enumerate}
\end{condition}

\noindent
Condition~\ref{cond:identification} places no restriction whatsoever on the density of $d_i$ at the origin. Part~\ref{id:limits} asks only that the conditional expectations have limits. Part~\ref{id:nomass} rules out a point mass of exact stayers; without it the trimmed limit identifies the movers' average partial effect rather than $\beta_0$. Part~\ref{id:bounded} imposes a uniform bound on the conditional expectations of $\beta(X_i)$ and $W_i^{*\prime}\beta(X_i)$ in a neighbourhood of the origin. This is a mild restriction: it is satisfied, for example, if $\beta(X_i)$ and $W_i^{*}$ are bounded. Note this condition is guaranteed by conditions in Section~\ref{sec:theory}.

The following proposition summarizes the identification result under a degenerate stayer density.

\begin{proposition}\label{prop:identification}
Under Condition~\ref{cond:identification}, $\delta_0$ and $\beta_0$ are identified by
\begin{align}\label{eq:id_degenerate}
    \delta_0 = \mathcal{A}(0^{+})^{-1}\mathcal{B}(0^{+})
    = \lim_{h \downarrow 0} \E[W_i^{*\prime}W_i^{*}\mid h_i=h]^{-1}\,
      \lim_{h \downarrow 0} \E[W_i^{*\prime}Y_i^{*}\mid h_i=h],
    \qquad
    \beta_0 = \lim_{h\downarrow0}\E\!\left[X_i^{-1}(Y_i-W_i\delta_0)\,\I[h_i\ge h]\right].
\end{align}
\end{proposition}

\begin{proof}[Proof of Proposition~\ref{prop:identification}]
\textbf{Step 1 ($\delta_0$).} Premultiplying (\ref{eq:adj_model}) by $W_i^{*\prime}$ and taking expectations conditional on $d_i=u$ gives, by iterated expectations and $\E[\epsilon_i\mid X_i]=0$,
\[
    \E\!\left[W_i^{*\prime}Y_i^{*}\mid d_i=u\right]
    = \E\!\left[W_i^{*\prime}W_i^{*}\mid d_i=u\right]\delta_0
      + u\,\E\!\left[W_i^{*\prime}\beta(X_i)\mid d_i=u\right].
\]
By Condition~\ref{cond:identification}.\ref{id:bounded}, the last term is bounded in norm by $C|u|$ for $0<|u|<\bar\delta$ and hence vanishes as $u\to0$. Letting $u\to0$ and using Condition~\ref{cond:identification}.\ref{id:limits} yields $\mathcal{B}(0^{+})=\mathcal{A}(0^{+})\delta_0$, and $\mathcal{A}(0^{+})$ is invertible.

\textbf{Step 2 (conditioning on $h_i$).} For $h>0$, the conditional law of $d_i$ given $h_i=h$ is supported on $\{-h,+h\}$; writing $\pi(h)=P(d_i=h\mid h_i=h)$, for any integrable $\mathcal{Y}$,
\[
    \E[\mathcal{Y}\mid h_i=h]
    = \pi(h)\,\E[\mathcal{Y}\mid d_i=h] + \left(1-\pi(h)\right)\E[\mathcal{Y}\mid d_i=-h].
\]
This is a convex combination of two quantities that share the same limit as $h\downarrow0$, by the two-sided limits in Condition~\ref{cond:identification}.\ref{id:limits}; it therefore converges to that common limit whatever the behaviour of $\pi(h)$. Applying this to $W_i^{*\prime}Y_i^{*}$ and to $W_i^{*\prime}W_i^{*}$ gives the second representation of $\delta_0$ in (\ref{eq:id_degenerate}).

\textbf{Step 3 ($\beta_0$).} On $\{h_i\ge h\}$ the matrix $X_i$ is invertible, and $\E[X_i^{-1}(Y_i-W_i\delta_0)\mid X_i]=\beta(X_i)$, so iterated expectations give
$\E[X_i^{-1}(Y_i-W_i\delta_0)\I[h_i\ge h]]=\E[\beta(X_i)\I[h_i\ge h]]$. By Condition~\ref{cond:identification}.\ref{id:nomass}, $\beta_0=\E[\beta(X_i)\I[h_i>0]]$, whence for $h<\bar\delta$
\[
    \normm{\E\!\left[\beta(X_i)\I[h_i\ge h]\right]-\beta_0}
    = \normm{\E\!\left[\E[\beta(X_i)\mid d_i]\,\I[0<h_i< h]\right]}
    \le C\,P(0<h_i<h)\;\longrightarrow\;0
\]
as $h\downarrow0$, using Condition~\ref{cond:identification}.\ref{id:bounded}. This is the second statement in (\ref{eq:id_degenerate}).
\end{proof}

% \section{Estimation}
% We follow a three-step procedure:

% \textbf{Step 1:} Estimate $\delta_0$ by Kernel regression when $T=p$:
% \begin{align}\label{eq:delta_estimation}
%     (\hat{\delta}_0,\dots,\hat{\delta}_{m}):=\min_{\gamma_0,\dots,\gamma_m} \sum_{i=1}^n \normm{Y_i^{*}-\sum_{k=0}^m h_i^k W_i^{*} \gamma_k}^2 K(h_i/h_0)
% \end{align}
% where $K(\cdot)$ is a kernel function with support $[0,1]$ and satisfies $\int_{0}^{1} K(u)du=1$. $h_0$ is the threshold below which the "Stayers" are defined.

% \textbf{Step 2:} Estimate $g_{\beta}(h)=\E[\beta(X_i)|h_i=h]$ of dimension $p$:
% \begin{align}\label{eq:g_estimation}
% \begin{aligned}
%     (\hat{\gamma}_0',\dots,\hat{\gamma}_m'):=&\min_{\gamma_0,\dots,\gamma_m} \sum_{i=1}^n \normm{\hat{\beta}(X_i)-\sum_{k=0}^m (h_i-h_0)^k \gamma_k}^2 K((h_i-h_0)/b)\\
%     \hat{g}_{\beta}(h_i) := &\sum_{k=0}^m \hat{\gamma}_k (h_i-h_0)^k
% \end{aligned}
% \end{align}
% where $\hat{\beta}(X_i) = (X_i'X_i)^{-1}X_i'(Y_i-W_i\hat{\delta_0})$ is an estimator of $\beta(X_i)$, $b$ is the bandwidth. We can choose $b=ch_0$, with $c>0$.

% \textbf{Step 3:} Estimate $\beta_0$ by
% \begin{align}\label{eq:beta0_estimation}
%     \hat{\beta}_0 = \frac{1}{n}\sum_{i=1}^n \hat{\beta}(X_i)\I[h_i\geq h_0] + \frac{1}{n}\sum_{i=1}^n \hat{g}_{\beta}(h_i)\I[h_i< h_0]
% \end{align}

% Note that letting $m=0$ and $K(\cdot)$ be the uniform kernel, we have the time drift estimator $\hat{\delta}_0 = \frac{1}{n}\sum_{i=1}^n W_i^{*\prime} Y_i^{*}\I[h_i< h_0]/ \frac{1}{n}\sum_{i=1}^n W_i^{*\prime}W_i^{*} \I[h_i< h_0]$, the same as the one in \cite{grahampowell}

\section{Estimation}\label{sec:estimation}
Define $\mathcal{A}(h)=\E[W_i^{*\prime}W_i^{*}\mid h_i=h]$, $\mathcal{B}(h)=\E[W_i^{*\prime}Y_i^{*}\mid h_i=h]$, $g_y(h)=\mathcal{A}(h)^{-1}\mathcal{B}(h)$, and $g_{\beta}(h) = \E[\beta(X_i)\mid h_i=h]$.
Proposition~\ref{prop:identification} identifies $\delta_0$ as the boundary value $g_y(0)$ of the
function $g_y(h)$, and $\beta_0$ as an average of $\beta(X_i)$ over a population with $h_i>0$.
This suggests a three-step procedure: estimate $g_y(0)$ by a local polynomial fit over the
stayers; use the resulting $\hat{\delta}_0$ to form unit-specific coefficients $\hat{\beta}(X_i)$
for the movers, and fit $\E[\beta(X_i)\mid h_i=h]$ near the threshold; then average,
imputing $\hat{g}_\beta(h_i)$ for the stayers. The steps are described in detail below.

\medskip\noindent
\textbf{Step 1: the aggregate time effects.}
Estimate $\delta_0$ by a local polynomial regression of $Y_i^{*}$ on $W_i^{*}$ at the boundary point $h=0$:
\begin{align}\label{eq:delta_estimation}
    (\hat{\delta}_0,\dots,\hat{\delta}_{m}):=\arg\min_{\gamma_0,\dots,\gamma_m}\ \sum_{i=1}^n \normm{Y_i^{*}-\sum_{k=0}^m h_i^k W_i^{*} \gamma_k}_2^2\, K(h_i/h_0),
    \qquad \gamma_k\in\R^{p(T-1)},
\end{align}
where $K$ is a kernel function that is nonnegative, supported on $[0,1]$, and normalised by $\int_0^1K(u)\,du=1$. Only units with $h_i\le h_0$ receive positive weight, so $h_0$ is the threshold below which the stayers are defined. By (\ref{eq:adj_model}), $\E[W_i^{*\prime}Y_i^{*}\mid h_i=h]=\mathcal{A}(h)g_y(h)$ with $g_y(0)=\delta_0$, so (\ref{eq:delta_estimation}) fits an $m$th-order Taylor expansion of $g_y$ at the origin. Only the intercept $\hat\delta_0$ estimates a parameter of interest; $\hat\delta_1,\dots,\hat\delta_m$ estimate $g_y^{(k)}(0)/k!$ and serve only to absorb the curvature of $g_y$ that would otherwise contaminate the intercept. Since the polynomial terms enter multiplied by $W_i^{*}$, (\ref{eq:delta_estimation}) is an ordinary weighted least squares fit of $Y_i^{*}$ on the $(m+1)p(T-1)$ regressors $(W_i^{*},h_iW_i^{*},\dots,h_i^mW_i^{*})$, and is computed as such.

\medskip\noindent
\textbf{Step 2: the conditional average coefficient near the boundary.}
For any unit with $h_i>0$ the matrix $X_i$ is invertible, and
\[
    \hat{\beta}(X_i) = (X_i'X_i)^{-1}X_i'\left(Y_i-W_i\hat{\delta}_0\right)=d_i^{-1}X_i^{*}\left(Y_i-W_i\hat\delta_0\right)
\]
estimates $\beta(X_i)$ up to the first-step error. Its precision deteriorates as $h_i\downarrow0$, which is why these units cannot simply be averaged. We instead use them to estimate $g_{\beta}(h)=\E[\beta(X_i)\mid h_i=h]$, a $p$-vector, by a second local polynomial fit centred at the threshold:
\begin{align}\label{eq:g_estimation}
\begin{aligned}
    (\hat{\gamma}_0',\dots,\hat{\gamma}_m'):=&\arg\min_{\gamma_0,\dots,\gamma_m}\ \sum_{i=1}^n \normm{\hat{\beta}(X_i)-\sum_{k=0}^m (h_i-h_0)^k \gamma_k}_2^2\, K\!\left((h_i-h_0)/b\right)\\
    \hat{g}_{\beta}(h) := &\sum_{k=0}^m \hat{\gamma}_k (h-h_0)^k,
\end{aligned}
\end{align}
where $b$ is a second bandwidth, tied to the first by $b=ch_0$ for some $c>0$. Since $K$ is supported on $[0,1]$, the fit uses only units with $h_i\in[h_0,(1+c)h_0]$: the movers closest to the threshold, for whom $\hat\beta(X_i)$ is noisiest but whose conditional mean is closest to the stayer region. The fitted polynomial is then evaluated at $h<h_0$, where no unit contributes to (\ref{eq:g_estimation}), so $\hat g_\beta$ is an extrapolation across the threshold, licensed by the smoothness of $g_\beta$.

\medskip\noindent
\textbf{Step 3: the average partial effect.}
Combine the two regions,
\begin{align}\label{eq:beta0_estimation}
    \hat{\beta}_0 = \frac{1}{n}\sum_{i=1}^n \hat{\beta}(X_i)\I[h_i\geq h_0] + \frac{1}{n}\sum_{i=1}^n \hat{g}_{\beta}(h_i)\I[h_i< h_0].
\end{align}
Movers contribute their own coefficient estimate; stayers, for whom $\hat\beta(X_i)$ is undefined or wildly imprecise, contribute the imputed value $\hat g_\beta(h_i)$. The estimator is thus an imputation estimator rather than a trimmed mean, and it targets $\E[\beta(X_i)]$ directly rather than through a limit.

\section{Asymptotic Theory}\label{sec:theory}
\subsection{Main result}

We first collect the notation used in the limit distributions. Define the unit-specific
regression error and the coefficient deviation
\[
    v_i=(X_i'X_i)^{-1}X_i'\epsilon_i=d_i^{-1}X_i^{*}\epsilon_i,
    \qquad
    \varsigma_i = \beta(X_i)-g_{\beta}(h_i),
\]
and the associated conditional second moments
\[
    \Sigma_v(h) = h^2\,var(v_i\mid h_i=h)=\E\!\left[X_i^{*}\epsilon_i\epsilon_i'X_i^{*\prime}\mid h_i=h\right],
    \qquad
    \Sigma_\beta(h) = var\!\left(\beta(X_i)\mid h_i=h\right)=\E[\varsigma_i\varsigma_i'\mid h_i=h].
\]
The scaling by $h^2$ in $\Sigma_v$ is what makes it bounded: $var(v_i\mid h_i=h)$ itself diverges
like $h^{-2}$ as $h\downarrow0$. For the first step, let
\[
    \zeta_i = Y_i^{*}-W_i^{*}g_y(h_i),\quad\Omega(h)=\E[W_i^{*\prime}\zeta_i\zeta_i'W_i^{*}\mid h_i=h]\in\R^{p(T-1)\times p(T-1)},
    \quad
    g_W(d) = \E[W_i^{*}\mid d_i=d]\in\R^{p\times p(T-1)}.
\]
Finally, write $\nu(u)=(1,u,\dots,u^m)'\in\R^{m+1}$, $\mathcal{U}(u)=\nu(u)\nu(u)'$, and let $e_1=(1,0,\dots,0)'\in\R^{m+1}$ be
the first unit vector of $\R^{m+1}$. The kernel constants entering the limits are
\begin{align}\label{eq:kernel_constants}
\begin{aligned}
    \mathcal{S}_m &= \alpha\int_{0}^{1}\mathcal{U}(u)K(u)u^{a}du,
    &\mathcal{T}_m &= \alpha\int_{0}^{1}\mathcal{U}(u)K^{2}(u)u^{a}du,
    &\mathcal{Q}_m &= \alpha\int_{0}^{1}\mathcal{U}(u)K(u)(1+cu)^{a}du,\\
    \tau_m &= \Big(\alpha(-1/c)^{k}B(k+1,a+1)\Big)_{k=0}^{m},
    &\omega_m &= \mathcal{Q}_m^{-1}\tau_m,
    &a_{+}(u) &= c^{-1}\nu(u)'K(u)\omega_m,
\end{aligned}
\end{align}
with $B(\cdot,\cdot)$ the Beta function, and
$\Xi_0 = \lim_{h\downarrow0}\E\!\left[d_i^{-1}W_i^{*}\I[h_i\ge h]\right]$ is the matrix through
which the first-step error propagates into $\hat\beta_0$. Assume the following conditions.

\begin{condition}\label{cond:Regularity}\leavevmode
    \begin{enumerate}
        \item\label{reg:iid} The data $\{(Y_i,X_i)\}_{i=1}^n$ are i.i.d., and $T=p$.
        \item\label{reg:bounded} $\normm{X_i}_F\leq C$ and $\normm{\beta(X_i)}_2\leq C$ almost surely.
              
        \item\label{reg:continuity} $\Sigma_v(h)$ and $\Sigma_\beta(h)$ are continuous at $h=0$.
        \item\label{reg:moment} $\E[\|\epsilon_i\|_2^{2+\delta}\mid X_i]\leq C$ for some $\delta>0$.
    \end{enumerate}
\end{condition}

\noindent
Condition~\ref{cond:Regularity}.\ref{reg:bounded} implies $\normm{W_i^{*}}_F\le C$, $var(\epsilon_i\mid X_i)=\sigma^2I_T$, $\normm{\Sigma_v(h)}_F\le C$, and $\normm{\Sigma_\beta(h)}_F\leq C$. It is an assumption of convenience; it can be replaced by moment bounds at the cost of
longer proofs. What it does \emph{not} deliver is boundedness of the conditional variance of
$\hat\beta(X_i)$ itself, which is of order $h_i^{-2}$ near the stayer boundary; this is the source
of the factor $h_0^{a-1}$ in the second-step variance. Condition~\ref{cond:Regularity}.\ref{reg:moment} supplies the Lyapunov condition for the triangular-array central limit theorems used in both steps.

\begin{condition}\label{cond:Smoothness}\leavevmode
    \begin{enumerate}
        \item\label{sm:density} The density of $d_i$ satisfies $f_D(d)=|d|^{a}q_D(d)$ with
              $q_D(0)=\alpha/2\neq 0$, $a>-1$, and $q_D\in C^{1}[-\bar{h},\bar{h}]$.
        \item\label{sm:functions} $\mathcal{A}(h)$, $\mathcal{B}(h)$, $\Omega(h)$ and $g_{\beta}(h)$
              belong to $C^{m+1}[0,\bar{h}]$, and $g_W(d)\in C^{1}[-\bar{h},\bar{h}]$. $\mathcal{A}(0^{+})$ exists and is invertible. $\mathcal{B}(0^{+})$ exists and is finite.
    \end{enumerate}
\end{condition}

\noindent
Condition~\ref{cond:Smoothness}.\ref{sm:density} implies that the density $f$ of $h_i=|d_i|$
satisfies $f(h)=f_D(h)+f_D(-h)=h^{a}\left(q_D(h)+q_D(-h)\right)$, so that $f(h)/h^{a}\to\alpha>0$
as $h\downarrow0$ and $f(h)\lesssim h^{a}$ on $(0,\bar h]$. The exponent $a$ is the single parameter governing how the stayer
population thins out at the boundary: $\Pr(h_i\le h)\asymp h^{a+1}$, so the first step is
effectively based on $nh_0^{a+1}$ observations. A larger $a$ means a sparser stayer population and
a less precise $\hat\delta_0$; a smaller $a$ means a denser one, and correspondingly more
observations with near-singular design matrices. The case
$a=0$, in which $f_D$ is bounded away from zero and infinity at the origin, is the configuration
assumed throughout \cite{grahampowell}. Condition~\ref{cond:Smoothness}.\ref{sm:functions} implies $g_y\in C^{m+1}[0,\bar h]$ and asks for $m+1$ derivatives, matching the order of the local polynomial fits; note that $g_W$ is required
to be smooth \emph{across} $d=0$, which is what licenses the $\pm d$ cancellation underlying the
existence of $\Xi_0$. Section~\ref{sec:primitive} gives primitive conditions on the joint density
of $(X_i,\epsilon_i)$ under which Condition~\ref{cond:Smoothness} holds.

\begin{condition}\label{cond:Kernel}\leavevmode
    \begin{enumerate}
        \item\label{ker:kernel} $K$ is nonnegative, once continuously differentiable, supported on
              $[0,1]$, normalised by $\int_0^1K(u)\,du=1$, and strictly positive
              on a subset of $[0,1]$ of positive Lebesgue measure.
        \item\label{ker:bandwidth} The bandwidths satisfy $b/h_0=c>0$.
        \item\label{ker:rates} $nh_0^{1-a}\to\infty$, $nh_0^{1+a}\to\infty$, and
              $nh_0^{a+2m+3}\to 0$.
    \end{enumerate}
\end{condition}

\noindent
The requirement that $K$ be strictly positive on a set of positive measure is what makes the moment
matrices $\mathcal{S}_m$ and $\mathcal{Q}_m$ in (\ref{eq:kernel_constants}) positive definite, since
each is a Gram matrix of the powers $1,u,\dots,u^m$ under a measure with density proportional to
$K$. The three requirements in Condition~\ref{cond:Kernel}.\ref{ker:rates} have distinct roles:
$nh_0^{1+a}\to\infty$ ensures that the stayer window contains enough units for the first step;
$nh_0^{1-a}\to\infty$ ensures that the near-boundary movers, whose errors are inflated by
$h_i^{-1}$, still vanish asymptotically; and $nh_0^{a+2m+3}\to0$ is the undersmoothing condition
$\sqrt{nh_0^{a+1}}\,h_0^{m+1}\to0$, which removes the bias of the local polynomial fit from the
limit distribution. Writing $h_0\asymp n^{-\gamma}$, the three hold jointly whenever
\[
    \frac{1}{a+2m+3}\;<\;\gamma\;<\;\frac{1}{1+|a|},
\]
an interval that is nonempty for every $a>-1$ and every $m\ge0$; in particular
$h_0\to0$ is implied and need not be assumed separately. Pushing $\gamma$ to its lower limit gives a
convergence rate approaching $n^{-\frac{1}{2}\left(1-\frac{1+a}{a+2m+3}\right)}$ for $a\ge 0$ and $n^{-\frac{1}{2}\left(1-\frac{1-a}{a+2m+3}\right)}$ for $a\le 0$, which at $a=0$
is $n^{-(m+1)/(2m+3)}$.

\begin{theorem}\label{theorem:asymptotic}
Suppose Conditions \ref{cond:Regularity}--\ref{cond:Kernel} hold. Let
$\phi_1 = e_1'\mathcal{S}_m^{-1}\mathcal{T}_m\mathcal{S}_m^{-1}e_1$,
$\Lambda_0 = \mathcal{A}(0)^{-1}\Omega(0)\mathcal{A}(0)^{-1}$, and, when $a<1$,
\begin{align}\label{eq:Upsilon_def}
    \Upsilon_0 := \alpha\,\varphi_0\,\Sigma_v(0),
    \qquad
    \varphi_0 := \int_{1}^{\infty}t^{a-2}
      \Big(1+a_{+}\big(\tfrac{t-1}{c}\big)\Big)^{2}dt
    = c\int_{0}^{1}(1+cu)^{a-2}\left(1+a_{+}(u)\right)^{2}du
      + \frac{(1+c)^{a-1}}{1-a}\;\in(0,\infty).
\end{align}
Then the estimators $\hat{\beta}_0$ and $\hat{\delta}_0$ defined in
equations~(\ref{eq:beta0_estimation}) and (\ref{eq:delta_estimation}) satisfy
the following.
\begin{enumerate}
    \item[(i)] \emph{(Case $a>0$: the first step dominates.)}
    \begin{align*}
        \sqrt{nh_0^{1+a}}\left(\hat{\delta}_0-\delta_0\right)
        \;\to_d\; N\!\left(0,\;\phi_1\Lambda_0\right),\quad \text{and} 
        \quad \sqrt{nh_0^{1+a}}\left(\hat{\beta}_0-\beta_0\right)
        \;\to_d\; N\!\left(0,\;\Xi_0\,\phi_1\Lambda_0\,\Xi_0'\right),
    \end{align*}
    and the joint limit is degenerate: with
    $\mathcal{Z}\sim N(0,\phi_1\Lambda_0)$,
    \begin{align*}
        \sqrt{nh_0^{1+a}}
        \begin{pmatrix}\hat{\delta}_0-\delta_0\\ \hat{\beta}_0-\beta_0\end{pmatrix}
        \;\to_d\;
        \begin{pmatrix}I_{p(T-1)}\\ -\Xi_0\end{pmatrix}\mathcal{Z}.
    \end{align*}
 
    \item[(ii)] \emph{(Case $a=0$: both steps contribute at the common rate
    $\sqrt{nh_0}$.)}
    \begin{align*}
        \sqrt{nh_0}
        \begin{pmatrix}\hat{\delta}_0-\delta_0\\ \hat{\beta}_0-\beta_0\end{pmatrix}
        \;\to_d\; N\!\left(0,\;
        \begin{pmatrix}
            \phi_1\Lambda_0 & -\phi_1\Lambda_0\Xi_0'\\
            -\Xi_0\phi_1\Lambda_0 & \Xi_0\phi_1\Lambda_0\Xi_0'+\Upsilon_0
        \end{pmatrix}\right).
    \end{align*}
 
    \item[(iii)] \emph{(Case $-1<a<0$: the two estimators converge at different
    rates and are asymptotically independent.)}
    \begin{align*}
        \begin{pmatrix}
            \sqrt{nh_0^{1+a}}\left(\hat{\delta}_0-\delta_0\right)\\[2pt]
            \sqrt{nh_0^{1-a}}\left(\hat{\beta}_0-\beta_0\right)
        \end{pmatrix}
        \;\to_d\; N\!\left(0,\;
        \begin{pmatrix}\phi_1\Lambda_0 & 0\\ 0 & \Upsilon_0\end{pmatrix}\right).
    \end{align*}
\end{enumerate}
\end{theorem}

\begin{remark}
The three cases trace out how the curvature parameter $a$ reallocates precision between the two
steps. When $a>0$ the stayer population is thin, so $\hat\delta_0$ is the imprecise ingredient and
its error dominates; $\hat\beta_0$ inherits the first-step limit distribution through $\Xi_0$, and
the two estimators are asymptotically perfectly dependent. When $a<0$ the stayer population is
thick, so $\hat\delta_0$ is comparatively well estimated, but the mover subsample is crowded with
near-singular design matrices and the second step becomes the binding constraint; the two
estimators then converge at different rates and are asymptotically independent. The case $a=0$
balances the two, and is the configuration studied by \cite{grahampowell}.

Setting $m=0$ and using the uniform kernel, our first step reduces to the aggregate-time-effect estimator of \cite{grahampowell}; if the imputation step is also omitted, the APE estimator reduces to their trimmed mover mean. Our higher-order local polynomial reduces the boundary bias from $O(h_0)$ to $O(h_0^{m+1})$, while the extrapolation step restores the contribution of the trimmed observations. At $a=0$, this changes the rate from the $n^{-1/3}$ benchmark toward $n^{-(m+1)/(2m+3)}$.

%\cite{Sasaki_2026} likewise use local polynomials to recover the contribution removed by trimming, but their correction is estimated in the signed determinant from the stayer and slow-mover window $|d_i|\le h_N$ and then added to the trimmed mover mean. The derivative terms are learned from nonzero slow movers in that window, while exact stayers enter through its mass and moments. We instead estimate $g_\beta(h)$ from movers just outside the threshold and extrapolate it into the stayer region before imputing each trimmed observation. Their theory permits exact stayers and a broad class of crowded or asymmetric continuous distributions; ours excludes a point mass but covers $f_D(d)=|d|^a q_D(d)$ for every $a>-1$, including the sparse case $a>0$, and gives an explicit rate decomposition across the three regimes.

\end{remark}

\subsection{Primitive Conditions for Smoothness}\label{sec:primitive}

Every random variable appearing in the model is a measurable function of
$(X_i,\epsilon_i)$, say $\mathcal{Y}(X_i,\epsilon_i)$. Since $X_i\in\R^{T\times p}$
has $Tp$ entries, we may reparametrise these entries as
$\left(x^{\bar p-1},d\right)$ with $\bar p:=T(p+1)$, where $d=d(X_i)=\det(X_i)$ is
the signed determinant and $x^{\bar p-1}\in\R^{\bar p-1}$ collects the remaining
coordinates together with $\epsilon_i$; any such variable is then written
$\mathcal{Y}\!\left(x^{\bar p-1},d\right)$.

Define the conditional expectation
\[
    g_{\mathcal{Y}}(u)\;:=\;\E\!\left[\mathcal{Y}\!\left(x^{\bar{p}-1},d_i\right)
      \,\middle|\,d_i=u\right]
    \;=\;\int \mathcal{Y}\!\left(x^{\bar{p}-1},u\right)\,
      f_{x^{\bar{p}-1}\mid d_i=u}\!\left(x^{\bar{p}-1}\mid u\right)\,
      d\!\left(x^{\bar{p}-1}\right).
\]
We impose the following conditions on a neighbourhood $[-h_{*},h_{*}]$ of the
origin.

\begin{condition}\label{cond:primitive}\leavevmode
\begin{enumerate}
\item\label{prim:factorization}
There exist $a>-1$ and a nonnegative function $q\!\left(x^{\bar{p}-1},u\right)$
such that the joint density of $\left(x^{\bar{p}-1},d(X_i)\right)$ satisfies
\[
    f_{x^{\bar{p}-1},d(X_i)}\!\left(x^{\bar{p}-1},u\right)
    \;=\; |u|^{a}\,q\!\left(x^{\bar{p}-1},u\right),
    \qquad u\in[-h_{*},h_{*}],
\]
and, for almost every $x^{\bar{p}-1}$, the map
$u\mapsto q\!\left(x^{\bar{p}-1},u\right)$ belongs to
$C^{m+1}\bigl([-h_{*},h_{*}]\bigr)$.

\item\label{prim:smoothY}
For almost every $x^{\bar{p}-1}$, the map
$u\mapsto \mathcal{Y}\!\left(x^{\bar{p}-1},u\right)$ belongs to
$C^{m+1}\bigl([-h_{*},h_{*}]\bigr)$, where $\mathcal{Y}$ represents elements of $\beta(X_i)$, $W_i^{*}$, $Y_i^{*}$.

\item\label{prim:envelope}
\[
\int
\underbrace{\Big(\max_{0\le k\le m+1}\ \sup_{u\in[-h_{*},h_{*}]}
\left|\partial_u^{k}\,q\!\left(x^{\bar{p}-1},u\right)\right|\Big)}_{=:\;\Theta_q\left(x^{\bar p-1}\right)}
\Big(1+
\underbrace{\max_{0\le k\le m+1}\ \sup_{u\in[-h_{*},h_{*}]}
\left|\partial_u^{k}\,\mathcal{Y}\!\left(x^{\bar{p}-1},u\right)\right|}_{=:\;\Theta_{\mathcal{Y}}\left(x^{\bar p-1}\right)}\Big)
d\!\left(x^{\bar{p}-1}\right)<\infty .
\]

\item\label{prim:nondegenerate}
\emph{Non-degeneracy at the boundary.}
\[
    \bar q(0)\;:=\;\int q\!\left(x^{\bar{p}-1},0\right)\,d\!\left(x^{\bar{p}-1}\right)\;>\;0 .
\]
\end{enumerate}
\end{condition}

\noindent
Condition~\ref{cond:primitive} restricts primitives only: the joint density of the reparametrised
data and the smoothness of $\mathcal{Y}$ in the determinant direction, $d_i$. Part~\ref{prim:factorization}
isolates the singular factor $|u|^{a}$ and requires what remains to be smooth across the origin,
which strengthens Condition~\ref{cond:Smoothness}.\ref{sm:density}. Part~\ref{prim:envelope} is the standard
hypothesis permitting the interchange of differentiation and integration, and
part~\ref{prim:nondegenerate} rules out the degenerate case for the function $q$. Note that no positivity of the density at $u=0$ is required: the factor $|u|^{a}$ may
vanish or diverge there.

\begin{lemma}\label{lemma:primitive_conditions}
Under Condition~\ref{cond:primitive} there exists $\bar h\in(0,h_{*}]$ such that
the following hold.
\begin{enumerate}
    \item The marginal density of $d(X_i)$ satisfies
    $f_{d(X_i)}(u)=|u|^{a}\bar q(u)$ with
    $\bar q\in C^{m+1}\bigl([-\bar h,\bar h]\bigr)$ and $\bar q(0)>0$; in
    particular $f_{d(X_i)}(u)/|u|^{a}\to \bar q(0)$ as $u\to 0$.
    \item $g_{\mathcal{Y}}\in C^{m+1}\bigl([-\bar h,\bar h]\bigr)$ and
    $\displaystyle\sup_{|u|\le \bar h}\bigl|g_{\mathcal{Y}}^{(k)}(u)\bigr|<\infty$
    for $k=0,1,\dots,m+1$, where $\mathcal{Y}$ represents elements of $\beta(X_i)$, $W_i^{*}$, $Y_i^{*}$.
\end{enumerate}
\end{lemma}

\begin{remark}
Part 1 delivers Condition~\ref{cond:Smoothness}.\ref{sm:density} with $q_D=\bar q$ and
$\alpha=2\bar q(0)$. Part 2 is stated on the \emph{two-sided} interval
$[-\bar h,\bar h]$, which is what Condition~\ref{cond:Smoothness}.\ref{sm:functions} requires of
$g_W$. The smoothness order is $m+1$
so as to match Condition~\ref{cond:Smoothness}.\ref{sm:functions}.
Taking $\mathcal{Y}$ scalar valued entails no loss: the vector- and
matrix-valued objects of interest --- $g_\beta$, $g_W$, $\mathcal{A}$,
$\mathcal{B}$ and $\Omega$ --- are obtained by applying the lemma to each entry
separately.

Finally, the objects conditioned on $h_i=|d_i|$ rather than on $d_i$ --- namely
$\mathcal{A}$, $\mathcal{B}$, $\Omega$ and $g_\beta$ --- are covered by the same
argument applied to the symmetrised numerator and denominator: for
$h\in(0,\bar h]$,
\[
    \E\!\left[\mathcal{Y}\mid h_i=h\right]
    = g_{\mathcal{Y}}(h)\frac{\bar q(h)}{\bar q(h)+\bar q(-h)} + g_{\mathcal{Y}}(-h)\frac{\bar q(-h)}{\bar q(h)+\bar q(-h)},
\]
because $f_{d(X_i)}(u)=|u|^{a}\bar q(u)$ assigns the two branches $u=\pm h$
weights proportional to $\bar q(h)$ and $\bar q(-h)$. Both $g_{\mathcal{Y}}$ and $\bar q$ are
$C^{m+1}$ on $[-\bar h,\bar h]$, and the denominator is bounded away from zero near the origin by
part (1), so $\E\!\left[\mathcal{Y}\mid h_i=h\right]\in C^{m+1}[0,\bar h]$.
\end{remark}

\section{Simulation study}\label{sec:simulations}
\subsection{A Simple Simulation Study}\label{sec:simplesim}
In this section, a known data-generating process is used to demonstrate the properties of the local linear estimator (we set $m=1$). We consider a short-panel setting with $T=2$ and a large cross-section, $N=10{,}000$. The regressors are generated as $X_{i1} \sim \mathcal{N}(0,1)$ and $X_{i2} \sim \mathcal{N}(1,1)$, and we define the key scalar $h_i = |X_{i2} - X_{i1}|$, which coincides with $|\det(X_i)|$ in the two-period case, consistent with the convention used throughout the paper. Individual coefficients are allowed to be correlated with regressors through $\beta_{i1} = 10h_i + \varepsilon_{i1}$ and $\beta_{i2} = 10h_i + \varepsilon_{i2}$, where $\varepsilon_{ij}$ are independent standard normal shocks. Outcomes are generated according to $Y_{i1} = \beta_{i1} + \beta_{i2}X_{i1} + u_{i1}$ and $Y_{i2} = \beta_{i1} + \beta_{i2}X_{i2} + u_{i2}$, with $u_{it}$ also standard normal. The parameter of interest is the average coefficient $\beta_0 = \mathbf{E}[\beta_i] = 10\,\mathbf{E}[h_i]\cdot(1,1)' \approx (14.00,14.00)'$. In each replication, the trimming threshold $h_0$ is set at the 10th percentile of $h_i$. The trimmed estimator of Graham and Powell (2012) averages the individual coefficients of the units with $h_i>h_0$. The local linear estimator fits $g_\beta(h)$ among those units with an Epanechnikov kernel, bandwidth $b = 30h_0$, and inverse-variance weights $h_i^2$, and imputes $\hat g_\beta(h_i)$ for the trimmed units. The simulation is repeated $5{,}000$ times to compute the mean squared error (MSE) and bias for both estimators, and Tables~\ref{tab:simpleresult} and~\ref{tab:varyingN} report them for the coefficient on $X_{it}$, $\beta_{i2}$. The true value $\beta_0$ is computed once by Monte Carlo from a draw of $2{,}000{,}000$ units. Note that, for a fair comparison with the Graham--Powell estimator, the density of $h_i$ is bounded away from zero at the origin, which corresponds to the case $a=0$ in the asymptotic theory.

\begin{table}[ht]
\centering
\caption{Finite Sample Performance ($N=10{,}000$)}
\label{tab:simpleresult}
\begin{tabular}{lcc}
\hline\hline
Estimator & MSE & Bias \\
\hline
Graham--Powell & 2.0383 & 1.4231 \\
Local Linear   & 0.0110 & $-0.0049$ \\
\hline\hline
\end{tabular}
\end{table}

Table~\ref{tab:simpleresult} reports results that reflect the ability of local linear smoothing to correct the first-order boundary bias that affects kernel estimators based on trimming. To further assess the finite-sample behavior of the two estimators, Table~\ref{tab:varyingN} reports results across cross-sectional sample sizes ranging from $N=250$ to $N=5000$ confirming that the local linear correction's bias elimination, not merely a variance improvement, is what drives its advantage.

\begin{table}[ht]
\centering
\caption{Finite Sample Performance Across Sample Sizes}
\label{tab:varyingN}
\begin{tabular}{lcccccc}
\hline\hline
& \multicolumn{3}{c}{Graham--Powell} & \multicolumn{3}{c}{Local Linear} \\
\cmidrule(lr){2-4} \cmidrule(lr){5-7}
$N$ & Bias & Var & MSE & Bias & Var & MSE \\
\hline
250   & 1.4232 & 0.5397 & 2.5650 & $-0.0003$ & 0.4535 & 0.4534 \\
500   & 1.4267 & 0.2622 & 2.2975 & \phantom{$-$}0.0002  & 0.2196 & 0.2196 \\
1000  & 1.4237 & 0.1292 & 2.1560 & $-0.0036$ & 0.1083 & 0.1083 \\
2000  & 1.4241 & 0.0631 & 2.0912 & $-0.0036$ & 0.0531 & 0.0531 \\
5000  & 1.4210 & 0.0254 & 2.0446 & $-0.0067$ & 0.0213 & 0.0213 \\
\hline\hline
\end{tabular}
\end{table}

\subsection{Empirical Study}

This section calibrates the empirical application of calorie demand among Nicaraguan households using the Red de Protección Social (RPS) balanced panel for 2000--2002. For $N$ households observed over $T=2$ periods, the outcome equation follows the correlated random-coefficient model
\begin{equation}\label{eq:sim_dgp}
    \log(\text{Cal}_{it}) = \beta_{i1}
        + \beta_{i2}\,\log(\text{Exp}_{it})
        + \varepsilon_{it},
\end{equation}
where $\text{Cal}_{it}$ is household calorie consumption and $\text{Exp}_{it}$ denotes weekly household food expenditure in
Nicaraguan Córdobas. The household-specific coefficients are generated as
\[\beta_{ij} = \beta_{0j}+ \lambda_j\,|\det(X_i)|+ \nu_{ij},\qquad j=1,2,\]
where $\det(X_i)=\log(\text{Exp}_{i1})-\log(\text{Exp}_{i2})$ and the $\nu_{ij}$ are independent normal shocks. We assume no time drift for now.

\subsection{Estimation Result}

\begin{table}[htbp]
\centering
\caption{OLS Regression of Individual Coefficient Estimates on the Determinant ($Y_0,Y_2$ Specification)}
\label{tab:determinant_regression}
\begin{tabular}{llccc}
\hline
Individual coefficient & Parameter & Estimate & $t$-statistic & $p$-value \\
\hline
$\hat{\beta}_{i1}$ (Intercept coefficient)
& Intercept ($\beta_{01}$)
& 6.3898
& 5.880
& 0.0000 \\

& Determinant slope ($\lambda_1$)
& $-3.8355$
& $-2.230$
& 0.0258\\

\hline

$\hat{\beta}_{i2}$ (Log expenditure coefficient)
& Intercept ($\beta_{01}$)
& 0.3143
& 2.820
& 0.0049 \\

& Determinant slope ($\lambda_2$)
& 0.3724
& 2.110
& 0.0348 \\

\hline
\end{tabular}

\vspace{0.2cm}
\begin{minipage}{0.95\textwidth}
\footnotesize
\textit{Notes:} The dependent variable in each regression is the estimated individual coefficient obtained from the Graham--Powell estimator. The explanatory variable is the determinant of the individual's regressor matrix, $\det(X_i)$. The regressions are estimated using only observations retained after trimming ($|\det(X_i)|>h$) at the 10th percentile. Both intercepts are significant at the 1\% level; both determinant slopes are significant at the 5\% level. In particular, the determinant has a significant negative effect on the estimated intercept coefficient ($\lambda_1<0$) and a significant positive effect on the estimated expenditure coefficient ($\lambda_2>0$). These results are obtained using the $Y_0,Y_2$ specification. The corresponding regressions using the $Y_1,Y_2$ specification do not yield statistically significant determinant effects (intercept coefficient: slope $-2.00$, $p=0.278$; log-expenditure coefficient: slope $0.18$, $p=0.349$).
\end{minipage}
\end{table}

These findings demonstrate that the determinant is not merely a numerical criterion for trimming observations but remains an economically important determinant of the individual coefficient estimates. Individuals whose regressor matrices are closer to singularity tend to exhibit systematically different estimated coefficients from those with larger determinants.

Importantly, this relationship is observed only when the estimator is constructed using the 2000 and 2002 observations $Y_0,Y_2$. Repeating the same exercise using the 2001 and 2002 observations $Y_1,Y_2$ does not produce a statistically significant relationship between the determinant and the estimated coefficients. This suggests that the influence of the determinant depends on the amount of within-individual variation in the regressors over the selected time periods. The larger time span between 2000 and 2002 generates greater variation in expenditure, leading to greater dispersion in the determinants and making their effect on the estimated coefficients more pronounced.

\subsection{Calibrated Simulation}
\label{sec:calibrated_sim}

The two specifications in Table~\ref{tab:determinant_regression} give two calibrated designs of exactly the form studied in Section~\ref{sec:simplesim}: one with a boundary slope that is statistically indistinguishable from zero ($Y_1,Y_2$), and one with a boundary slope that is significant ($Y_0,Y_2$). We use each, in turn, to calibrate $(\beta_{0j},\lambda_j)$ in the random-coefficient equation~\eqref{eq:sim_dgp} and resample household expenditure pairs from the RPS data, preserving their empirical joint distribution rather than imposing a parametric form.

\begin{table}[htbp]
\centering
\caption{Calibrated finite-sample performance, intercept coefficient}
\label{tab:calibrated_sim}
\begin{tabular}{lcccccc}
\hline\hline
& \multicolumn{3}{c}{Graham--Powell} & \multicolumn{3}{c}{Local Polynomial} \\
\cmidrule(lr){2-4} \cmidrule(lr){5-7}
$N$ & Bias & Var & MSE & Bias & Var & MSE \\
\hline
\multicolumn{7}{l}{\emph{Panel A: $Y_0,Y_2$ calibration (significant boundary slope)}} \\
250   & $-0.2284$ & 2.3515 & 2.4032 & $-0.0478$ & 2.5247 & 2.5264 \\
500   & $-0.1982$ & 1.1702 & 1.2093 & $-0.0155$ & 1.2551 & 1.2551 \\
1000  & $-0.2139$ & 0.5991 & 0.6447 & $-0.0316$ & 0.6392 & 0.6401 \\
2000  & $-0.2182$ & 0.3052 & 0.3527 & $-0.0374$ & 0.3266 & 0.3280 \\
5000  & $-0.2116$ & 0.1160 & 0.1607 & $-0.0266$ & 0.1218 & 0.1225 \\
\hline
\multicolumn{7}{l}{\emph{Panel B: $Y_1,Y_2$ calibration (flat boundary slope)}} \\
250   & $-0.1280$ & 2.3138 & 2.3297 & $-0.0328$ & 2.4220 & 2.4226 \\
500   & $-0.0554$ & 1.1756 & 1.1785 & \phantom{$-$}0.0248  & 1.2449 & 1.2453 \\
1000  & $-0.0669$ & 0.5842 & 0.5885 & \phantom{$-$}0.0211  & 0.6163 & 0.6166 \\
2000  & $-0.0674$ & 0.2999 & 0.3044 & \phantom{$-$}0.0189  & 0.3148 & 0.3151 \\
5000  & $-0.0787$ & 0.1156 & 0.1218 & \phantom{$-$}0.0068  & 0.1224 & 0.1224 \\
\hline\hline
\end{tabular}

\vspace{0.4em}
\begin{minipage}{0.92\textwidth}
\footnotesize \emph{Notes.} $5{,}000$ Monte Carlo replications per cell. $X_i$ resampled with replacement from the real $(\log\text{Exp}_{i,a},\log\text{Exp}_{i,b})$ pairs in the RPS panel; $(\beta_{0j},\lambda_j,\sigma_{\nu_j})$ from the corresponding row of Table~\ref{tab:determinant_regression}. Trimming and bandwidth as in Table~\ref{tab:varyingN}.
\end{minipage}
\end{table}

Table~\ref{tab:calibrated_sim} makes the contrast plain. Under the $Y_0,Y_2$ calibration, where the CMC genuinely has a slope to correct, the Graham--Powell estimator carries a persistent bias of about $-0.21$ to $-0.22$ across the entire sample-size range, while the local polynomial estimator's bias is smaller by roughly a factor of five to eight throughout and falls as low as $-0.015$. Under the $Y_1,Y_2$ calibration, by contrast, the Graham--Powell bias is itself much smaller, so the local polynomial estimator has very little to correct. The two conditions under which the estimators behave differently are the density of stayers and the curvature of $g_\beta(h)$ near the boundary. When the boundary slope is flat, the two estimators perform similarly; when the boundary slope is significant, the local polynomial estimator's bias correction is substantial and its MSE advantage is clear.

\newpage

% =====================================================================
%  Section for the main paper: T = p = 5 application to the full ENIA
%  census.  Requires: booktabs, amsmath, graphicx.
%  Figures: fig_scatter.pdf, fig_ghat.pdf
% =====================================================================

% =====================================================================
%  Section for the main paper: T = p = 5 application to the full ENIA
%  census.  Requires: booktabs, amsmath, graphicx.
%  Figures: fig_scatter.pdf, fig_ghat.pdf
% =====================================================================

% =====================================================================
%  Section for the main paper: T = p = 4 application to the ENIA census.
%  Requires: booktabs, amsmath, graphicx.
%  Figures: fig_scatter.pdf, fig_ghat.pdf
% =====================================================================

% =====================================================================
%  Section for the main paper: T = p = 4 application to the ENIA census.
%  Requires: booktabs, amsmath, graphicx.
%  Figures: fig_scatter.pdf, fig_ghat.pdf
% =====================================================================

% =====================================================================
%  Section for the main paper: T = p = 4 application to the ENIA census.
%  Requires: booktabs, amsmath, graphicx.
%  Figures: fig_scatter.pdf, fig_ghat.pdf
% =====================================================================

\section{Application: heterogeneous input elasticities in Chilean
manufacturing}
\label{sec:chile}

\subsection{Specification}

Production function estimation is the canonical setting in which a coefficient such as output elasticity
that is almost certainly heterogeneous is nevertheless restricted to be common
across units.\footnote{The typical model in industrial organization is \begin{equation}
    y_{i} = \beta_L l_i + \beta_K k_i +  w_i + \varepsilon_i  
\end{equation}
where $y$ is output, $l$ is labour, $k$ is capital, $w$ is productivity observed by firms but not by econometricians, and $\varepsilon$ is a random shock.} \citet{Raval2023}, for example, studies tests of firm markups using multiple flexible inputs. Methodologically, theoretical papers including \citet{OlleyPakes1996}, \citet{LevinsohnPetrin2003}, \citet{ACF2015}, and \citet{gandhi2020identification} use various identification assumptions to address the correlation between endogenous productivity and input choices, but nevertheless assume common elasticities.

Several papers allow production technologies to vary across firms. \citet{MairesseGriliches1990} estimate firm-specific production functions from long panels, while \citet{KasaharaSchrimpfSuzuki2023} model a finite mixture of latent production technologies. \citet{LiSasaki2024} identify firm-specific flexible-input elasticities from ex-ante cost shares, and \citet{Li2026TERC} estimates a heterogeneous Cobb--Douglas production function with time-varying endogenous random coefficients. We allow the productivity intercept and all input elasticities to differ across plants and to be correlated with their input histories, while a common time drift allows these coefficients to change over time. Specifically, we set up the gross-output production-function model
\begin{equation}
  y_{it} \;=\; X_{it}'(\beta_i + \delta_t) \;+\; \varepsilon_{it},
  \qquad
  X_{it} \;=\; (1,\; k_{it},\; l_{it},\; x_{it})',
  \qquad t = 1,\dots,4,
  \label{eq:klm}
\end{equation}
where $y$ is log real gross output, $k$ and $l$ are logs of the real capital
stock and employment, and $x$ is the log of real intermediate inputs,
constructed as the sum of deflated materials and energy. This is the
single-intermediate-input form used by \citet{gandhi2020identification}; aggregating materials
and energy is not a concession here but a choice we return to in
Section~\ref{sec:designsize}, since it materially improves what the design can
resolve.

Note that unobserved productivity is represented by the coefficient on the constant regressor. It, along with all
elasticities, may vary across plants and be correlated with the input
history. The vector $\delta_t$ is a common drift that shifts the entire
coefficient vector in period $t$, allowing the technology itself to move over
the window while $\beta_i$ carries the permanent plant-specific component. We
normalise $\delta_1 = 0$, so the estimand $\mathbf{E}[\beta_i]$ is the average
elasticity vector anchored at the first year of the window.

\subsection{Data and design}

We use the full Encuesta Nacional Industrial Anual, the annual census of
Chilean manufacturing plants with at least ten employees, covering 1979-1996:
86{,}186 plant-year observations on 10{,}927 plants. Real gross output, real
materials, real energy and the real capital stock (buildings, machinery and
vehicles, mid-year) are sectorally deflated series; the capital stocks are
those of \citet{Greenstreet2007}, built by a permanent inventory procedure
separately for structures, equipment and transportation. Labour is the total
employment headcount. We access the census through the replication package of
\citet{Raval2023}, and depart from his construction in using the mid-year real
capital stock rather than a rental-cost aggregate, since the auxiliary
depreciation and interest-rate files his code merges are not needed for our
purpose.

Estimating \eqref{eq:klm} requires plants observed in four \emph{consecutive}
years with strictly positive output and all inputs. Coverage is remarkably flat
across the sample period, with every window from 1987 onwards retaining between
3{,}250 and 3{,}414 plants; we use 1992--1995, which retains the most
(3{,}414) and sits in the post-stabilisation period rather than spanning the
debt crisis. Because each plant contributes a single observation, the units are
independent and no clustering is required. Table~\ref{tab:design} reports the
design diagnostics of Section~2.

\begin{table}[t]
\centering
\caption{Design diagnostics, ENIA 1992--1995}
\label{tab:design}
\begin{tabular}{lc}
\toprule
Plants ($N$), each observed $T = 4$ years        & 3{,}414 \\
Coefficients ($p$)                               & 4 \\
Drift parameters ($Q = (T-1)p$)                  & 12 \\
$\pi_0 = \Pr(\det X_i = 0)$ (exact atom)         & 0.0076 \\
Median $h_i$                                     & $5.97 \times 10^{-3}$ \\
$h_i$ at the 10th percentile                     & $5.05 \times 10^{-4}$ \\
Median $\operatorname{cond}(X_i)$                & $5.8 \times 10^{3}$ \\
$\hat a$ in $f(h) \sim h^{a}$                    & $-0.077$ \\
\bottomrule
\end{tabular}

\vspace{0.6em}
\begin{minipage}{0.92\textwidth}
\footnotesize \emph{Notes.} $\hat a$ is the slope of $\log \hat f(h)$ on
$\log h$ over the lowest three deciles of $h$, from a reflected Epanechnikov
density estimate.
\end{minipage}
\end{table}

Two features matter. First, $\hat a \approx 0$: the density of the determinant
is bounded away from zero at the boundary, placing it in the thick-boundary regime $a < 1$ of
Theorem~1 in which the trimmed estimator converges more slowly than $\sqrt N$.
Second, the point mass at $\det X_i = 0$ is negligible. Requiring a
$4 \times 4$ design to be exactly singular is far more demanding than requiring
one input to be unchanged between two adjacent years, and fewer than $1\%$ of plants
satisfy it. Assumption~4(d) is therefore not in question, and the trimming
proportion is a free choice rather than something the data forces.

We use that freedom to keep it small. Every estimate in Table~\ref{tab:main}
trims at most $10\%$ of plants. The bias of the trimmed estimator is
proportional to $p_0$, so holding $p_0$ below one tenth keeps the estimand
close to the average partial effect over the whole population rather than over
an endogenously selected subpopulation of movers. The costs are the ones noted
above---retained ill-conditioned designs, and a thinner stayer sample for the
drift step---and we let the standard errors display them rather than buying
precision by trimming a quarter of the sample.

\subsection{Is there heterogeneity to speak of?}

Equation~\eqref{eq:klm} is exactly identified, so within-plant residuals vanish
and the model cannot test its own homogeneity restriction. We add a fifth year,
1996, which leaves 3{,}158 plants with one residual degree of freedom apiece.
Estimating $\delta$ by pooled least squares on the interacted design, purging
it from the outcome and then fitting plant-specific regressions gives
$\hat\sigma = 0.130$, and the \citet{Swamy1970} statistic for
$H_0: \beta_i = \beta$ for all $i$ is $S = 1.19 \times 10^{5}$ on $12{,}580$
degrees of freedom. Homogeneity is rejected overwhelmingly. The dispersion is
economically large as well: among movers the interquartile range of
$\hat\beta_i$ is $1.19$ for intermediates and $2.15$ for capital, and $0.93$
and $1.39$ respectively in the better-conditioned half of the sample where
estimation error contributes less.

Heterogeneity alone does not require the machinery of this paper; only
heterogeneity \emph{correlated} with the determinant does. We regress the
estimated individual coefficients on $h_i$ among movers,
\begin{equation}
  \hat\beta_{i,p} \;=\; \alpha_p + \kappa_p h_i + e_{ip},
  \label{eq:kappa}
\end{equation}
weighting by $h_i^2$, which is proportional to the inverse conditional variance
when the $\Sigma_v(h_i)/h_i^{2}$ term dominates. Unweighted least squares is
uninformative here, being dominated by near-boundary observations that carry
almost no information about $g$.

\begin{table}[t]
\centering
\caption{Selection gradient $\hat\kappa_p$ in \eqref{eq:kappa}}
\label{tab:kappa}
\begin{tabular}{lccc}
\toprule
& $\hat\kappa$ & $t$ & $p$ \\
\midrule
Capital       & $-0.040$ & $-0.82$ & $0.410$ \\
Labour        & $-0.635$ & $-5.13$ & $<0.001$ \\
Intermediates & $\phantom{-}0.621$ & $\phantom{-}3.00$ & $0.003$ \\
\bottomrule
\end{tabular}

\vspace{0.6em}
\begin{minipage}{0.92\textwidth}
\footnotesize \emph{Notes.} Movers only, trimming at $10\%$
($N_{\text{mov}} = 3{,}072$, $342$ stayers). Weighted by $h_i^2$;
heteroskedasticity-robust standard errors.
\end{minipage}
\end{table}

Labour and intermediates show sharp gradients of opposite sign; capital does
not. Plants whose input bundle moves least have markedly \emph{higher}
estimated labour elasticities and lower intermediate elasticities. This is the
configuration in which the trimmed estimator is biased and an extrapolation
estimator is called for.

Figure~\ref{fig:scatter} plots the individual coefficient estimates against
$h_i$ with the trimmed region shaded, and Figure~\ref{fig:ghat} the local
linear estimate of $g$ with its extrapolation to the boundary. Both display the
$h_i^{-2}$ heteroskedasticity plainly: the cloud fans out sharply as $h_i$
approaches the threshold and the confidence band for $\hat g$ widens in the
same region. The labour panel of Figure~\ref{fig:ghat} shows the gradient of
Table~\ref{tab:kappa} directly, $\hat g$ rising steeply as $h \to h_0$.

\begin{figure}[t]
\centering
\includegraphics[width=\textwidth]{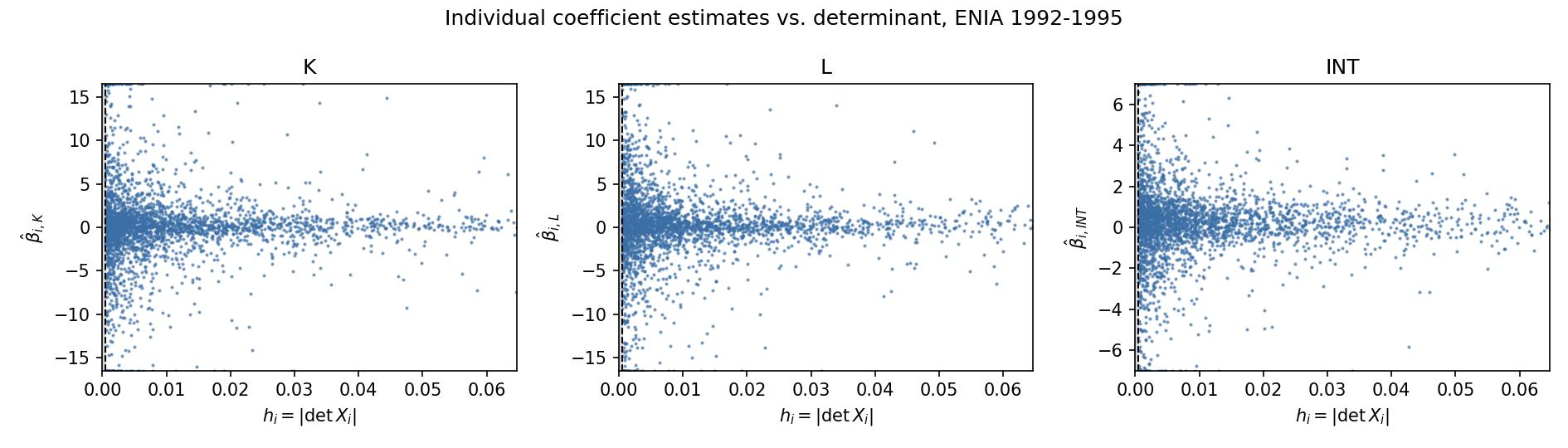}
\caption{Individual coefficient estimates $\hat\beta_{i,p}$ against
$h_i = |\det X_i|$, ENIA 1992--1995. Shaded region: units trimmed at
$h_0 = q_{10}(h)$. Vertical axes truncated at the 97th percentile of
$|\hat\beta_{i,p}|$.}
\label{fig:scatter}
\end{figure}

\begin{figure}[t]
\centering
\includegraphics[width=\textwidth]{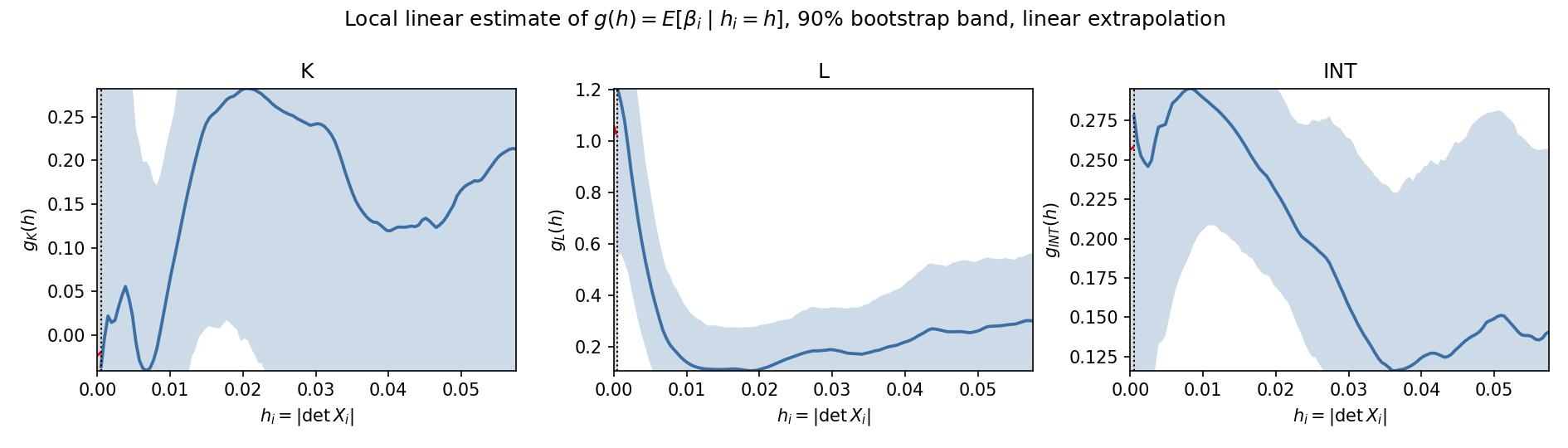}
\caption{Local polynomial estimate of $g(h) = \mathbf{E}[\beta_i \mid h_i = h]$
among movers (solid), with 90\% bootstrap band, and its linear extrapolation
into the trimmed region (dashed). Shaded region as in Figure~\ref{fig:scatter}.}
\label{fig:ghat}
\end{figure}

\subsection{The estimated drift}
\label{sec:drift}

Table~\ref{tab:drift} reports $\hat\delta_t$ at the $10\%$ threshold, where
$342$ stayers identify the twelve parameters, comparing the pooled,
unweighted ($m=0$) fit of \citet{grahampowell} with the local-linear
($m=1$) drift estimator of Section~\ref{sec:estimation}.

\begin{table}[t]
\centering
\caption{Estimated coefficient drift $\hat\delta_t$, $10\%$ trimming
(1992 normalised to zero)}
\label{tab:drift}
\begin{tabular}{lcccc}
\toprule
& Intercept & Capital & Labour & Intermediates \\
\midrule
\multicolumn{5}{l}{\emph{Panel A: pooled drift, $m=0$}} \\
1993 & $-0.459$ & $-0.065$ & $\phantom{-}0.008$ & $0.091$ \\
     & (0.260) & (0.030) & (0.063) & (0.039) \\
1994 & $-1.270$ & $-0.108$ & $-0.082$ & $0.227$ \\
     & (0.471) & (0.049) & (0.112) & (0.067) \\
1995 & $-1.630$ & $-0.107$ & $-0.169$ & $0.286$ \\
     & (0.584) & (0.062) & (0.148) & (0.086) \\
\midrule
\multicolumn{5}{l}{\emph{Panel B: local polynomial drift, $m=1$}} \\
1993 & $-0.948$ & $-0.103$ & $-0.055$ & $0.193$ \\
     & (0.446) & (0.054) & (0.115) & (0.084) \\
1994 & $-1.348$ & $-0.147$ & $-0.153$ & $0.298$ \\
     & (0.992) & (0.085) & (0.182) & (0.165) \\
1995 & $-1.784$ & $-0.189$ & $-0.207$ & $0.386$ \\
     & (1.224) & (0.102) & (0.243) & (0.200) \\
\bottomrule
\end{tabular}

\vspace{0.6em}
\begin{minipage}{0.92\textwidth}
\footnotesize \emph{Notes.} $\hat\delta$ from \eqref{eq:delta_estimation}
with a uniform kernel, estimated on the $342$ units with
$h_i \leq h_0 = q_{10}(h)$. Bootstrap standard errors over plants, 400
replications. Wald test of $H_0:\delta=0$: Panel A, $24.05$ on twelve
degrees of freedom ($p=0.020$); Panel B, $14.80$ ($p=0.253$).
\end{minipage}
\end{table}

The point estimates are directionally the same under both orders, but the
Wald test is not: the pooled ($m=0$) fit rejects $H_0:\delta=0$ comfortably
at the $5\%$ level, while the local polynomial ($m=1$) correction, with
roughly double the standard errors throughout, no longer rejects at
conventional levels. The pooled fit's rejection was evidently borrowing precision from the same kind
of first-order trimming bias documented elsewhere in this paper, and once
that is corrected the $342$-stayer sample is simply not large enough to pin
the drift down as tightly. We now report this as a suggestive
pattern rather than a statistically established one, and note that it is
the drift step, not the coefficient step, that makes low trimming
expensive.

\subsection{Estimates}

\begin{table}[t]
\centering
\caption{Average input elasticities, ENIA 1992--1995}
\label{tab:main}
\begin{tabular}{lcccc}
\toprule
Trimmed & Stayers & Capital & Labour & Intermediates \\
\midrule
\multicolumn{5}{l}{\emph{Panel A: Graham--Powell trimmed estimator}} \\
\phantom{0}2\% & \phantom{0}69 & $0.038$ & $0.649$ & $0.225$ \\
               &     & (0.910) & (0.848) & (0.476) \\
\phantom{0}4\% & 137 & $0.279$ & $0.539$ & $0.253$ \\
               &     & (0.362) & (0.503) & (0.238) \\
\phantom{0}6\% & 205 & $0.184$ & $0.535$ & $0.239$ \\
               &     & (0.288) & (0.367) & (0.177) \\
\phantom{0}8\% & 274 & $0.217$ & $0.367$ & $0.206$ \\
               &     & (0.218) & (0.292) & (0.151) \\
10\%           & 342 & $0.242$ & $0.222$ & $0.306$ \\
               &     & (0.214) & (0.269) & (0.140) \\
\midrule
\multicolumn{5}{l}{\emph{Panel B: local polynomial correction}} \\
\phantom{0}2\% & \phantom{0}69 & $0.037$ & $0.773$ & $0.213$ \\
\phantom{0}4\% & 137 & $0.280$ & $0.509$ & $0.269$ \\
\phantom{0}6\% & 205 & $0.187$ & $0.539$ & $0.243$ \\
\phantom{0}8\% & 274 & $0.189$ & $0.429$ & $0.206$ \\
10\%           & 342 & $0.214$ & $0.324$ & $0.303$ \\
\midrule
\multicolumn{5}{l}{\emph{Panel C: common-coefficient benchmarks}} \\
Pooled OLS, 1992 base & & $0.112$ & $0.224$ & $0.715$ \\
Pooled OLS, year mean & & $0.112$ & $0.238$ & $0.704$ \\
Within, 1992 base     & & $0.044$ & $0.257$ & $0.486$ \\
\bottomrule
\end{tabular}

\vspace{0.6em}
\begin{minipage}{0.92\textwidth}
\footnotesize \emph{Notes.} $N = 3{,}414$ plants, $T = p = 4$. Local polynomial
estimator uses $m = 1$, $c = b/h_0 = 30$ and inverse-variance weights, and the
drift step uses the same $m=1$ local-linear correction as Table~\ref{tab:drift}
(both panels share this drift estimate, subtracted before the coefficient
step). Standard errors from 400 bootstrap replications over plants; Panel~B
estimates carry standard errors within $0.02$ of the corresponding Panel~A
entries and are omitted. Panel~C estimates
$y_{it} = X_{it}'(\beta + \delta_t)$ by pooled and within least squares. All
thresholds are at or below $10\%$ by design.
\end{minipage}
\end{table}

The intermediate-input elasticity is the coefficient this design resolves
most precisely. It sits between $0.21$ and $0.31$ across the $4$--$10\%$
grid, with standard errors between $0.14$ and $0.24$, and lies well below
both the pooled OLS value of $0.70$ and the within value of $0.49$ at every
threshold. Allowing the elasticity to be plant-specific and correlated with
input history therefore lowers the average substantially relative to the
common-coefficient estimate, which is the kind of quantity the model exists
to deliver, and which no amount of proxy-variable machinery would recover,
since the proxies are aimed at the intercept.

Capital and labour are estimated markedly less precisely once the drift
step's own boundary bias is corrected with standard errors of $0.21$--$0.27$ at
the $10\%$ threshold, and considerably wider still at $2\%$, reflect
honestly the cost of the thinner stayer sample the drift step relies on.
Both coefficients remain positive and broadly in the neighbourhood of the
benchmark values from $4\%$ trimming onward. This reason for being less precise is that the conditional variance of $\hat\beta_i$ is governed by the
corresponding row of $\operatorname{adj}(X_i)$, and intermediate inputs are
adjusted far more freely across years than capital or headcount, so their
coefficient is the one the determinant can resolve most sharply.

The $2\%$ row should still be read with care. It trims so little that only
$69$ stayers remain to identify 12 drift parameters, and the resulting
standard errors are roughly over $0.5$ for every coefficient. We suggest using results with $4\%$ trimming as benchmark. 

Finally, under correlated coefficient the average partial effect and the
within estimator do not estimate the same object. We conduct a Wald test to test the hypothesis that they (two population mean) are the same. Bootstrapping the difference at the $10\%$
threshold, the capital and intermediates differences are within $1.25$
standard errors of zero ($t=0.93$ and $t=-1.20$), and the labour
difference is negligible ($t=-0.13$); the joint Wald statistic is $3.27$ on
three degrees of freedom ($p = 0.35$). Given the standard errors in
Table~\ref{tab:main} we argue that the interpretation of this result should not be the
absence of correlated heterogeneity. Homogeneity is decisively rejected from the sharp
gradients in Table~\ref{tab:kappa} for two of three inputs, and the
point estimate for intermediates differs from the within value by $0.18$. However, since the average partial effect and the within estimator represents some population mean of the random heterogeneous coefficients, it is possible that in our application and design they are the same.

\subsection{Markups and the overidentification test}
\label{sec:markup}

The Introduction motivated this application by an overidentification test in the spirit of \citet{Raval2023}. Under cost minimization, each variable input's first-order condition delivers its own estimate of the same firm-level markup, $\text{markup}_j = \theta_j / s_j$, where $\theta_j$ is that input's output elasticity and $s_j$ is its share of revenue.\footnote{Revenue shares require nominal, undeflated variables that play no role in \eqref{eq:klm} itself: nominal gross output for revenue, and nominal expenditure on each input. From the underlying ENIA records we construct revenue from nominal gross output, wage as the sum of wages, bonuses, payroll and family-allowance taxes, and materials expenditure as the sum of nominal raw materials purchases, nominal fuel purchases and electricity purchased. The share is taken at the base year of the window (1992), matching the $\delta_1=0$ normalisation under which $\hat\beta$ is literally the elasticity vector in year one. Averaging over the $3{,}395$ of $3{,}414$ plants with a valid, positive share that year gives a mean labour share of $14.9\%$ and a mean materials share of $52.0\%$.} If the maintained production-function assumptions are correct, the markups implied by different inputs should agree. We now carry that test out, using labour and intermediates as the two different inputs.

\begin{table}[t]
\centering
\caption{Input-specific markups, ENIA 1992--1995}
\label{tab:markup}
\begin{tabular}{lcccccc}
\toprule
Trimmed & \multicolumn{3}{c}{Graham--Powell} & \multicolumn{3}{c}{Local polynomial} \\
\cmidrule(lr){2-4}\cmidrule(lr){5-7}
& Labour & Materials & Diff. & Labour & Materials & Diff. \\
\midrule
2\%  & 4.367 & 0.433 & 3.934 & 5.202 & 0.409 & 4.793 \\
4\%  & 3.629 & 0.487 & 3.142 & 3.426 & 0.518 & 2.908 \\
6\%  & 3.602 & 0.459 & 3.143 & 3.625 & 0.468 & 3.157 \\
8\%  & 2.470 & 0.395 & 2.075 & 2.885 & 0.395 & 2.489 \\
10\% & 1.496 & 0.587 & 0.908 & 2.181 & 0.582 & 1.599 \\
\bottomrule
\end{tabular}

\vspace{0.6em}
\begin{minipage}{0.92\textwidth}
\footnotesize \emph{Notes.} $\text{markup}_j = \hat\theta_j / s_j$ using the Panel A and Panel B elasticities of Table~\ref{tab:main} and the base-year (1992) revenue shares above ($s_L=0.149$, $s_M=0.520$).
\end{minipage}
\end{table}

Table~\ref{tab:markup} shows the labour-implied markup running three to ten times the materials-implied markup at every trimming threshold, and driven mainly by labour's small revenue share amplifying an already imprecisely estimated output elasticity into a large implied markup. Whether this gap is more than sampling noise is a separate question, which we address by bootstrapping $H_0:\text{markup}_L=\text{markup}_M$ over $400$ plant resamples, redrawing both the elasticities and the revenue shares from the same resampled plants each time.

\begin{table}[t]
\centering
\caption{Overidentification test, $H_0: \text{markup}_L=\text{markup}_M$}
\label{tab:overid}
\begin{tabular}{lcccccc}
\toprule
Trimmed & \multicolumn{3}{c}{Graham--Powell} & \multicolumn{3}{c}{Local polynomial} \\
\cmidrule(lr){2-4}\cmidrule(lr){5-7}
& Diff. & SE & $p$ & Diff. & SE & $p$ \\
\midrule
2\%  & 3.934 & 5.925 & 0.507 & 4.793 & 6.183 & 0.438 \\
4\%  & 3.142 & 3.720 & 0.398 & 2.908 & 3.814 & 0.446 \\
6\%  & 3.143 & 2.637 & 0.233 & 3.157 & 2.660 & 0.235 \\
8\%  & 2.075 & 2.103 & 0.324 & 2.489 & 2.176 & 0.253 \\
10\% & 0.908 & 1.879 & 0.629 & 1.599 & 1.883 & 0.396 \\
\bottomrule
\end{tabular}
\end{table}

None of the tests in Table~\ref{tab:overid} rejects $H_0$ at conventional levels. Running the identical exercise under the homogeneous-elasticity restriction, using pooled OLS elasticities of $0.224$ (labour) and $0.715$ (intermediates), gives markups of $1.509$ and $1.374$ and a much smaller point gap of $0.135$ (bootstrap SE $0.100$, $t=1.35$, $p=0.176$), also a non-rejection.

The comparison across the two rows is the point worth making. Allowing elasticities to be plant-specific and correlated with input history does \emph{not} shrink the labour/materials markup gap relative to the homogeneous benchmark, but the estimate is also not statistically sharper because of larger standard errors, particularly for the labour elasticity. Thus, neither version of the test has much power here. We report this as an honest, if inconclusive, finding rather than force a rejection that the data do not support.

\subsection{Design size choice}
\label{sec:designsize}

The specification in \eqref{eq:klm} aggregates materials and energy into a
single intermediate input. Disaggregating them raises the design $T = p$ from  4 to 5, the drift
parameters $Q$ from 12 to 20, and lowers $N$ from 3{,}414 to 3{,}133 and the stayer count at
$p_0 = 0.10$ from $342$ to $314$. Estimating the larger model on the same
census and at the same thresholds, greatly inflates the variance of obtained estimators, due to only $63$ stayers being used to estimate 20 drift parameters.

\subsection{Discussion}
\label{sec:discussion}

The local polynomial correction plays two different roles in this application. In the drift step, the Wald test of a stable
technology across 1993-95 is no longer rejected at conventional levels once the drift's own
boundary bias is removed (Table~\ref{tab:drift}), exactly the kind of reversal the simulations of
Section~\ref{sec:simulations} anticipate. In the overidentification test of Section~\ref{sec:markup},
by contrast, neither the trimmed nor the local polynomial estimator rejects equal markups across inputs because the labour elasticity is imprecisely estimated
once heterogeneity is allowed for.

The more informative comparison in this application is not Graham-Powell against local polynomial
but the correlated-random-coefficients estimate, in either form, against the homogeneous-elasticity
benchmark the markup literature has used until now. No existing implementation of the production
approach to markup estimation has previously allowed the elasticities
entering the markup formula to be anything other than common across producers. Once they are allowed
to vary, the homogeneous model's own markup gap turns out to be smaller than either
correlated-random-coefficients estimate by roughly an order of magnitude
(Table~\ref{tab:overid} versus the pooled benchmark).

\section{Conclusion}
 
This paper has studied identification and estimation of the average partial effect in a correlated
random coefficients panel data model with $T=p$, the configuration in which the parameter is
irregularly identified and the information bound is singular. Our aim was to weaken the conditions
under which such identification is available, and to convert smoothness of the underlying
conditional expectations into a faster rate of convergence than the trimmed-mean estimator of
\cite{grahampowell} attains.
 
Three findings emerge. First, identification does not require the determinant of the design matrix
to have a density bounded away from zero at the origin. Replacing the boundary evaluation by an
explicit boundary limit, both the aggregate time effects and the average partial effect are
identified under weak conditions. Second, the exponent $a$ governs
how precision is allocated between the two estimation steps, and produces three distinct limiting
regimes: for $a>0$ the sparse stayer population makes the time effects the binding ingredient and
the two estimators are asymptotically perfectly dependent; for $a<0$ the crowding of near-singular
design matrices among the movers binds instead, and the two converge at different rates and are
asymptotically independent; $a=0$ balances the two. Third, the major bias in
\cite{grahampowell} arises from discarding the stayers, not from smoothing. Because the conditional
mean coefficient is smooth up to the boundary $0$, it can be estimated from the movers and extrapolated
across the threshold, which removes that bias rather than shrinking it away and relaxes the
undersmoothing requirement from $nh_0^{3}\to0$ to $nh_0^{a+2m+3}\to0$.
 
The practical advantage is a convergence rate that improves with the polynomial order, approaching
the parametric rate as $m$ grows, at the cost only of higher-order smoothness of objects that are
already assumed differentiable.

The Chilean application confirms that this diagnosis is not a modelling convenience. Correlated
heterogeneity is economically large (Table~\ref{tab:kappa}), and correcting the drift step reverses a
substantive conclusion about the stability of the technology across years. Although the overidentification
test built on it does not distinguish the trimmed estimator from the local polynomial correction, the idea to be conveyed by this paper is that no existing implementation of the
production approach to markup estimation has allowed elasticity heterogeneity at all, and doing so
here moves the estimated markup gap by roughly an order of magnitude relative to the
homogeneous-elasticity benchmark. 

Several extensions remain. The polynomial order and the bandwidth ratio $c$ are treated as fixed;
a data-driven choice of both, and of $h_0$ itself, would be valuable. Finally, the
extensions considered by \cite{grahampowell}, overidentification when $T>p$, a point mass of
exact stayers, and event study style discrete valued regressors with correlated random coefficients, can plausibly be
accommodated within the present framework.

\begin{appendices}
\section{Lemmas}
\begin{lemma}[Lindeberg--Feller central limit theorem Proposition 2.27 in \cite{van2000asymptotic}]\label{lemma:LF_CLT}
For each $n$ let $Y_{n,1}, \dots, Y_{n,k_n}$ be independent random vectors with finite variances such that
\begin{align}
\sum_{i=1}^{k_n} \mathrm{E} \|Y_{n,i}\|^2 \mathbf{1}\{\|Y_{n,i}\| > \varepsilon\} &\to 0, \qquad \text{every } \varepsilon > 0, \\
\sum_{i=1}^{k_n} \operatorname{Cov} Y_{n,i} &\to \Sigma.
\end{align}
Then the sequence $\sum_{i=1}^{k_n} (Y_{n,i} - \mathrm{E} Y_{n,i})$ converges in distribution to $N(0,\Sigma)$.
\end{lemma}

\begin{lemma}\label{lemma:delta_convergence}
Suppose Conditions \ref{cond:Regularity}--\ref{cond:Kernel} hold and $T=p$. Then the estimator $\hat\delta_0$ defined in
equation~(\ref{eq:delta_estimation}) satisfies
\begin{align*}
    \hat{\delta}_0-\delta_0
    = O_p\!\left(h_0^{m+1}\right)
    + o_p\!\left((nh_0^{a+1})^{-1/2}\right)
    + \left(e_1'\mathcal{S}_m^{-1}\otimes \mathcal{A}(0)^{-1}\right) V_n ,
\end{align*}
where $V_n:=\frac{1}{nh_0^{a+1}}\sum_{i=1}^n Z_i'\zeta_i K(h_i/h_0)$, $Z_i = (W_i^*,\frac{h_i}{h_0}W_i^*,\dots,(\frac{h_i}{h_0})^m W_i^*)\in\R^{T\times(m+1)p(T-1)}$, and
\begin{align*}
    \sqrt{nh_0^{a+1}}\left(e_1'\mathcal{S}_m^{-1}\otimes \mathcal{A}(0)^{-1}\right) V_n
    \;\to_d\;
    N\!\left(0,\;\left(e_1'\mathcal{S}_m^{-1}\mathcal{T}_m\mathcal{S}_m^{-1}e_1\right)
    \cdot\left(\mathcal{A}(0)^{-1}\Omega(0)\mathcal{A}(0)^{-1}\right)\right).
\end{align*}
Here $e_1=(1,0,\dots,0)'\in\R^{m+1}$,
\begin{align*}
    \mathcal{S}_m = \alpha\int_{0}^{1}\mathcal{H}(u)\,K(u)\,u^{a}\,du,
    \qquad
    \mathcal{T}_m = \alpha\int_{0}^{1}\mathcal{H}(u)\,K^2(u)\,u^{a}\,du,
    \qquad
    \mathcal{H}(u) = \left(u^{k+l-2}\right)_{1\le k,l\le m+1},
\end{align*}
$\zeta_i = Y_i^{*}-W_i^{*}g_y(h_i)$, and
$\Omega(h)=\E\!\left[W_i^{*\prime}\zeta_i\zeta_i'W_i^{*}\mid h_i=h\right]
\in\R^{p(T-1)\times p(T-1)}$.
\end{lemma}

\begin{lemma}\label{lemma:beta0_positive}
Suppose Conditions \ref{cond:Regularity}--\ref{cond:Kernel} hold. Then the
principal-value limit
\begin{align}\label{eq:Xi0_def}
    \Xi_0 \;:=\; \lim_{h\downarrow0}\
    \E\!\left[d_i^{-1}W_i^{*}\,\I[\,|d_i|\ge h\,]\right]
    \;=\;\lim_{h\downarrow0}\ \E\!\left[(X_i'X_i)^{-1}X_i'W_i\,\I[h_i\ge h]\right]
\end{align}
exists and is finite, and the estimator defined in
equation~(\ref{eq:beta0_estimation}) satisfies
\begin{align*}
    \E_n \hat{\beta}(X_i)\I[h_i\geq h_0] - \E \beta(X_i)\I[h_i\geq h_0]
    = \E_n v_i\I[h_i\geq h_0] - \Xi_0\Delta\delta
      + o_p\!\left(\normm{\Delta\delta}_2\right) + O_p\!\left(n^{-1/2}\right),
\end{align*}
where $\Delta\delta=\hat\delta-\delta_0$ and
$v_i=(X_i'X_i)^{-1}X_i'\epsilon_i$.
\end{lemma}

\begin{remark}
The limit in (\ref{eq:Xi0_def}) must be read as a principal value: for $a\le0$
the map $d\mapsto d^{-1}g_W(d)f_D(d)$ is not absolutely integrable near the
origin, and finiteness of $\Xi_0$ rests on the cancellation between $+d$ and
$-d$ established in Step~2 below. This is exactly where the $C^1$ smoothness of
$g_W$ and $q_D$ across $d=0$ in Condition~\ref{cond:Smoothness} is used.
\end{remark}

\begin{lemma}\label{lemma:beta0_negative}
Suppose Conditions \ref{cond:Regularity}--\ref{cond:Kernel} hold. Then the
estimator defined in equation~(\ref{eq:beta0_estimation}) satisfies
\begin{align*}
    \E_n \hat{g}_\beta(h_i)\I[h_i<h_0] - \E g_\beta(h_i)\I[h_i<h_0]
    = \Lambda_n^{v}\,\omega_m
      + o_p\!\left(\sqrt{h_0^{a-1}/n}\right)
      + O_p\!\left(h_0^{a+1}\right)\Delta\delta
      + O_p\!\left(h_0^{m+a+2}\right)
      + o_p\!\left(n^{-1/2}\right),
\end{align*}
where
\begin{alignat*}{2}
    \Lambda_n^{v} &= \frac{1}{nc}\sum_{i=1}^n v_i\,
        \nu\!\left(\tfrac{h_i-h_0}{b}\right)'K\!\left(\tfrac{h_i-h_0}{b}\right)
    \in\R^{p\times(m+1)},
    &\qquad
    \omega_m &= \mathcal{Q}_m^{-1}\tau_m\in\R^{m+1},\\
    \mathcal{Q}_m &=\alpha\int_{0}^{1}\mathcal{H}(u)K(u)(1+cu)^{a}\,du,
    &\qquad
    \tau_{m,k}&=\alpha\left(-\tfrac{1}{c}\right)^{k}B(k+1,a+1),\\
     \tau_m&=(\tau_{m,0},\dots,\tau_{m,m})',
\end{alignat*}
 The leading term satisfies
$\E\left[\Lambda_n^{v}\omega_m\right]=0$ and
$\Var\!\left(\Lambda_n^{v}\omega_m\right) =O\!\left(h_0^{a-1}/n\right)$.
\end{lemma}

\begin{proof}[Proof of Lemma~\ref{lemma:primitive_conditions}]
 
Throughout write
\[
    \bar q(u)=\int q\!\left(x^{\bar p-1},u\right)d\!\left(x^{\bar p-1}\right),
    \qquad
    N_q(u)=\int \mathcal{Y}\!\left(x^{\bar p-1},u\right)
      q\!\left(x^{\bar p-1},u\right)d\!\left(x^{\bar p-1}\right),
\]
both integrals being finite by
Condition~\ref{cond:primitive}(\ref{prim:envelope}).
 
\textbf{Step 0: reduction to a ratio.}
Integrating the factorisation in
Condition~\ref{cond:primitive}(\ref{prim:factorization}) over $x^{\bar p-1}$
gives
\[
    f_{d(X_i)}(u)=|u|^{a}\,\bar q(u),\qquad u\in[-h_{*},h_{*}],
\]
and therefore
\[
    f_{x^{\bar p-1}\mid d_i=u}\!\left(x^{\bar p-1}\mid u\right)
    =\frac{f_{x^{\bar p-1},d(X_i)}\left(x^{\bar p-1},u\right)}{f_{d(X_i)}(u)}
    =\frac{q\!\left(x^{\bar p-1},u\right)}{\bar q(u)}
\]
at every $u$ with $\bar q(u)>0$. The singular factor $|u|^{a}$ cancels between
the joint and the marginal density: although both are singular at $u=0$
whenever $a\ne0$, the conditional density is not, and it is this cancellation
that allows $g_{\mathcal{Y}}$ to inherit the smoothness of $q$ and
$\mathcal{Y}$. Consequently
\begin{align}\label{eq:g_ratio}
    g_{\mathcal{Y}}(u)=\frac{N_q(u)}{\bar q(u)}
\end{align}
at every such $u$; Step 2 will show that this holds on a neighbourhood of the
origin.
 
\textbf{Step 1: $N_q$ and $\bar q$ are $C^{m+1}$ on $[-h_{*},h_{*}]$.}
Fix $k\in\{0,1,\dots,m+1\}$. For almost every $x^{\bar p-1}$,
Condition~\ref{cond:primitive}(\ref{prim:factorization}) gives
$u\mapsto q\!\left(x^{\bar p-1},u\right)\in C^{m+1}$ and
Condition~\ref{cond:primitive}(\ref{prim:smoothY}) gives
$u\mapsto \mathcal{Y}\!\left(x^{\bar p-1},u\right)\in C^{m+1}$, so the integrand
of $N_q$ is $k$ times continuously differentiable in $u$, with
\[
\partial_u^{k}\!\left[\mathcal{Y}\!\left(x^{\bar p-1},u\right)
  q\!\left(x^{\bar p-1},u\right)\right]
=\sum_{j=0}^{k}\binom{k}{j}\,
\partial_u^{j}\mathcal{Y}\!\left(x^{\bar p-1},u\right)\,
\partial_u^{k-j}q\!\left(x^{\bar p-1},u\right)
\]
by the Leibniz rule. Uniformly in $u\in[-h_{*},h_{*}]$, the right-hand side is
bounded in absolute value by
\[
    2^{k}\,\Theta_{\mathcal{Y}}\!\left(x^{\bar p-1}\right)
    \Theta_{q}\!\left(x^{\bar p-1}\right)
    \;\le\;2^{m+1}\,\Theta_{q}\!\left(x^{\bar p-1}\right)
    \left(1+\Theta_{\mathcal{Y}}\!\left(x^{\bar p-1}\right)\right),
\]
which is integrable in $x^{\bar p-1}$ by
Condition~\ref{cond:primitive}(\ref{prim:envelope}). Iterating for $k=1,\dots,m+1$, and
applying dominated convergence once more to obtain continuity of the resulting
integrals, yields $N_q\in C^{m+1}\bigl([-h_{*},h_{*}]\bigr)$ with
\[
N_q^{(k)}(u)
=\sum_{j=0}^{k}\binom{k}{j}\int
\partial_u^{j}\mathcal{Y}\!\left(x^{\bar p-1},u\right)\,
\partial_u^{k-j}q\!\left(x^{\bar p-1},u\right)\,d\!\left(x^{\bar p-1}\right).
\]
Taking $\mathcal{Y}\equiv 1$ — by
Condition~\ref{cond:primitive}(\ref{prim:envelope}) — the identical argument gives
$\bar q\in C^{m+1}\bigl([-h_{*},h_{*}]\bigr)$ with
\[
    \bar q^{(k)}(u)=\int \partial_u^{k}q\!\left(x^{\bar p-1},u\right)
    d\!\left(x^{\bar p-1}\right).
\]
Together with $f_{d(X_i)}(u)=|u|^{a}\bar q(u)$ from Step 0 and
Condition~\ref{cond:primitive}(\ref{prim:nondegenerate}), this proves part (i)
for any $\bar h\in(0,h_{*}]$.
 
\textbf{Step 2: the denominator is bounded away from zero near the origin.}
By Condition~\ref{cond:primitive}(\ref{prim:nondegenerate}) we have
$\bar q(0)>0$, and by Step 1 the function $\bar q$ is continuous on
$[-h_{*},h_{*}]$. Hence there exist $\bar h\in(0,h_{*}]$ and $\eta>0$ such that
\[
    \bar q(u)\ \ge\ \eta\ >\ 0\qquad\text{for all }|u|\le \bar h .
\]
In particular $f_{d(X_i)}(u)>0$ for $0<|u|\le \bar h$, so $g_{\mathcal{Y}}$ is
well defined there and the representation (\ref{eq:g_ratio}) is valid on
$[-\bar h,\bar h]$.
 
\textbf{Step 3: the ratio is $C^{m+1}$ with bounded derivatives.}
On $[-\bar h,\bar h]$ the function $g_{\mathcal{Y}}=N_q/\bar q$ is a quotient of
two $C^{m+1}$ functions (Step 1) whose denominator does not vanish (Step 2).
Since the reciprocal of a nonvanishing $C^{m+1}$ function is $C^{m+1}$, and
products of $C^{m+1}$ functions are $C^{m+1}$, it follows that
$g_{\mathcal{Y}}\in C^{m+1}\bigl([-\bar h,\bar h]\bigr)$. Each derivative
$g_{\mathcal{Y}}^{(k)}$, $k=0,1,\dots,m+1$, is then continuous on the compact
interval $[-\bar h,\bar h]$ and hence bounded there, which proves part (ii).
\end{proof}

\begin{proof}[Proof of Lemma \ref{lemma:delta_convergence}]
    \textbf{Step 0 (Reparametrisation).}
    Let $\nu_i = \nu(h_i/h_0) = \left(1,\tfrac{h_i}{h_0},\dots,(\tfrac{h_i}{h_0})^m\right)'
    \in\R^{m+1}$ and
    \begin{align*}
        Z_i \;=\; \nu_i'\otimes W_i^{*}
        \;=\;\left(W_i^{*},\ \tfrac{h_i}{h_0}W_i^{*},\ \dots,\ (\tfrac{h_i}{h_0})^m W_i^{*}\right)
        \;\in\;\R^{p\times (m+1)p(T-1)} ,
    \end{align*}
    recalling that $W_i^{*}=X_i^{*}W_i\in\R^{p\times p(T-1)}$ and $T=p$. With this
    notation the minimisation problem (\ref{eq:delta_estimation}) is the weighted
    least squares problem
    \begin{align}\label{eq:theta_hat_def}
        \hat{\theta} \;=\; \arg\min_{\theta}\ \sum_{i=1}^n
        \normm{Y_i^{*}-Z_i\theta}_2^2\,K(h_i/h_0).
    \end{align}

    \textbf{Step 1 (Orthogonality of the local residual).}
    Let $\zeta_i = Y_i^{*}-W_i^{*}g_y(h_i)$. Since
    $g_y(h)=\mathcal{A}(h)^{-1}\mathcal{B}(h)$ with
    $\mathcal{A}(h)=\E[W_i^{*\prime}W_i^{*}\mid h_i=h]$ and
    $\mathcal{B}(h)=\E[W_i^{*\prime}Y_i^{*}\mid h_i=h]$, we obtain
    \begin{align}\label{eq:orthogonality}
        \E\!\left[W_i^{*\prime}\zeta_i \mid h_i\right]
        = \mathcal{B}(h_i)-\mathcal{A}(h_i)\mathcal{A}(h_i)^{-1}\mathcal{B}(h_i)=0 .
    \end{align}
    This orthogonality condition — and not the stronger, and in general false,
    statement $\E[\zeta_i\mid h_i]=0$ — is what the argument below requires. It
    implies $\E\!\left[Z_i'\zeta_i\mid h_i\right]=\nu_i\otimes\E[W_i^{*\prime}\zeta_i\mid h_i]=0$
    and hence $\E V_n=0$.

    \textbf{Step 2 (A preliminary moment bound).}
    For every $q\ge 0$, using $\mathrm{supp}(K)=[0,1]$, the change of variable
    $h=uh_0$ and Condition~\ref{cond:Smoothness} gives,
    \begin{align}\label{eq:kernel_moment}
        \E\left|K(h_i/h_0)\right|^{q}
        = \int_0^{h_0}\left|K(h/h_0)\right|^{q}f(h)\,dh
        = h_0\int_0^1 \left|K(u)\right|^{q}f(uh_0)\,du
        \;\lesssim\; h_0^{a+1}\int_0^1 \left|K(u)\right|^{q}u^{a}\,du
        \;\lesssim\; h_0^{a+1}.
    \end{align}
    In particular $P(h_i\le h_0)\asymp h_0^{a+1}$, so that $nh_0^{a+1}$ is the
    effective number of observations entering (\ref{eq:theta_hat_def}); this
    explains the normalisation adopted below.

    \textbf{Step 3 (Decomposition).}
    By Condition \ref{cond:Smoothness} a Taylor expansion of $g_y$ at $0$ gives
    \begin{align*}
        g_y(h_i) = \sum_{k=0}^m \frac{g_y^{(k)}(0)}{k!}h_i^k + R_{y}(h_i),
        \qquad \normm{R_{y}(h_i)}_2\lesssim h_i^{m+1},
    \end{align*}
    uniformly for $h_i\in[0,h_0]$. Consequently
    \begin{align*}
        Y_i^{*} = W_i^{*}g_y(h_i) + \zeta_i
        = \sum_{k=0}^m h_i^k W_i^{*}\frac{g_y^{(k)}(0)}{k!} + W_i^{*}R_{y}(h_i) + \zeta_i
        = Z_i\theta_0 + W_i^{*}R_{y}(h_i) + \zeta_i,
    \end{align*}
    where
    \begin{align*}
        \theta_0 = \left(\frac{h_0^{0}g_y(0)'}{0!},\ \frac{h_0^{1}g_y^{(1)}(0)'}{1!},\
        \dots,\ \frac{h_0^{m}g_y^{(m)}(0)'}{m!}\right)'\in\R^{(m+1)p(T-1)},
    \end{align*}
    the $h_0^k$ factors compensating for the rescaling $h_i^k=(h_i/h_0)^k h_0^k$
    built into $Z_i$. Since $\delta_0=g_y(0)$ is the leading block of $\theta_0$,
    \begin{align}\label{eq:delta_estimation_equivalence}
        \hat{\delta}_0-\delta_0
        = \left(e_1'\otimes I_{p(T-1)}\right)(\hat{\theta}-\theta_0).
    \end{align}
    Define
    \begin{align*}
        S_n &= \frac{1}{nh_0^{a+1}}\sum_{i=1}^n Z_i'Z_i\,K(h_i/h_0)
            = \frac{1}{nh_0^{a+1}}\sum_{i=1}^n
            \left(\nu_i\nu_i'\otimes W_i^{*\prime}W_i^{*}\right)K(h_i/h_0),\\
        \Lambda_n &= \frac{1}{nh_0^{a+1}}\sum_{i=1}^n Z_i'W_i^{*}R_{y}(h_i)\,K(h_i/h_0),
        \qquad
        V_n = \frac{1}{nh_0^{a+1}}\sum_{i=1}^n Z_i'\zeta_i\,K(h_i/h_0).
    \end{align*}
    On the event that $S_n$ is invertible — which by Step 4 occurs with probability
    tending to one — the first-order condition of (\ref{eq:theta_hat_def}) yields
    \begin{align}\label{eq:theta_decomposition}
        \hat{\theta}-\theta_0
        = \left[\sum_{i=1}^n Z_i'Z_i K(h_i/h_0)\right]^{-1}
        \left[\sum_{i=1}^n Z_i'\left(W_i^{*}R_{y}(h_i)+\zeta_i\right)K(h_i/h_0)\right]
        = S_n^{-1}\Lambda_n + S_n^{-1}V_n .
    \end{align}

    \textbf{Step 4 (Convergence of $S_n$).}
    Since the observations are i.i.d.,
    $\E S_n = h_0^{-(a+1)}\,\E\!\left[\left(\nu_i\nu_i'\otimes \mathcal{A}(h_i)\right)K(h_i/h_0)\right]$
    by the law of iterated expectations. Substituting $h=uh_0$,
    \begin{align*}
        \E S_n
        &= \frac{1}{h_0^{a+1}}\int_{0}^{h_0}
        \left(\mathcal{H}(h/h_0)\otimes\mathcal{A}(h)\right)K(h/h_0)f(h)\,dh\\
        &= \frac{1}{h_0^{a}}\int_{0}^{1}
        \left(\mathcal{H}(u)\otimes\mathcal{A}(uh_0)\right)K(u)f(uh_0)\,du\\
        &= \int_{0}^{1}\left(\mathcal{H}(u)\otimes\mathcal{A}(uh_0)\right)K(u)\,u^{a}\,
        \frac{f(uh_0)}{(uh_0)^{a}}\,du
        \;\longrightarrow\;
        \alpha\int_{0}^{1}\mathcal{H}(u)K(u)u^{a}\,du \otimes \mathcal{A}(0),
    \end{align*}
    where the limit follows from $h_0\to0$, the continuity of $\mathcal{A}$ at $0$,
    $f(uh_0)/(uh_0)^{a}\to\alpha$, and dominated convergence (the integrand is
    bounded by a constant multiple of $K(u)u^{a}$). Hence
    \begin{align}\label{eq:delta_Sn_convergence}
        \E S_n = \mathcal{S}_m\otimes \mathcal{A}(0) + o(1).
    \end{align}
    For the variance, i.i.d.\ sampling, Condition \ref{cond:Regularity},
    $\normm{\nu_i}_2\lesssim 1$ on $\{h_i\le h_0\}$, the identity
    $\normm{A\otimes B}_F=\normm{A}_F\normm{B}_F$ and (\ref{eq:kernel_moment}) give
    \begin{align}\label{eq:delta_Sn_variance}
        \E\normm{S_n-\E S_n}_F^2
        \le \frac{1}{n h_0^{2(a+1)}}\,\E\normm{Z_i'Z_i K(h_i/h_0)}_F^2
        = \frac{1}{n h_0^{2(a+1)}}\,
        \E\left[\normm{\nu_i\nu_i'}_F^2\normm{W_i^{*\prime}W_i^{*}}_F^2 K(h_i/h_0)^2\right]
        \lesssim \frac{1}{nh_0^{a+1}}\to 0 .
    \end{align}
    Combining (\ref{eq:delta_Sn_convergence}) and (\ref{eq:delta_Sn_variance}),
    $S_n\to_p \mathcal{S}_m\otimes \mathcal{A}(0)$.
    
    The limit is invertible. Indeed, $\mathcal{S}_m$ is the Hankel (moment) matrix
    of the finite measure $\alpha K(u)u^{a}du$ on $[0,1]$, and for any
    $c=(c_1,\dots,c_{m+1})'\ne 0$,
    \begin{align*}
        c'\mathcal{S}_m c
        = \alpha\int_0^1\Big(\sum_{k=1}^{m+1}c_k u^{k-1}\Big)^2 K(u)\,u^{a}\,du > 0,
    \end{align*}
    because a non-zero polynomial has finitely many roots while $K(u)u^{a}>0$ on a
    set of positive measure. Thus $\mathcal{S}_m\succ 0$; together with
    $\mathcal{A}(0)\succ0$ this makes
    $\mathcal{S}_m\otimes\mathcal{A}(0)$ positive definite, and $S_n^{-1}=O_p(1)$.
    
    \textbf{Step 5 (Bias term).}
    By the triangle inequality, Condition \ref{cond:Regularity},
    $\normm{R_y(h_i)}_2\lesssim h_i^{m+1}\le h_0^{m+1}$ on the support of
    $K(\cdot/h_0)$, and (\ref{eq:kernel_moment}),
    \begin{align*}
        \E\normm{\Lambda_n}_2
        \;\le\; \frac{1}{h_0^{a+1}}\,\E\normm{Z_i'W_i^{*}R_{y}(h_i)}_2 K(h_i/h_0)
        \;\lesssim\; \frac{h_0^{m+1}}{h_0^{a+1}}\,\E\,K(h_i/h_0)
        \;\lesssim\; h_0^{m+1}.
    \end{align*}
    Markov's inequality therefore gives
    \begin{align}\label{eq:delta_Lambda_convergence}
        \Lambda_n = O_p\!\left(h_0^{m+1}\right),
        \qquad\text{and hence}\qquad
        S_n^{-1}\Lambda_n = O_p\!\left(h_0^{m+1}\right).
    \end{align}
    
    \textbf{Step 6 (Asymptotic normality of $V_n$).}
    By (\ref{eq:orthogonality}), $\E V_n=0$. Using
    $(\nu_i\otimes A)(\nu_i'\otimes B)=\nu_i\nu_i'\otimes AB$ and iterated
    expectations,
    \begin{align*}
        nh_0^{a+1}\operatorname{Var}(V_n)
        &= \frac{1}{h_0^{a+1}}\,
        \E\left[\left(\mathcal{H}(h_i/h_0)K(h_i/h_0)^2\right)
        \otimes \left(W_i^{*\prime}\zeta_i\zeta_i'W_i^{*}\right)\right]\\
        &= \frac{1}{h_0^{a+1}}\int_{0}^{h_0}
        \left(\mathcal{H}(h/h_0)K(h/h_0)^2\right)\otimes\Omega(h)\,f(h)\,dh\\
        &= \int_{0}^{1}\left(\mathcal{H}(u)K(u)^2\right)\otimes\Omega(uh_0)\,u^{a}\,
        \frac{f(uh_0)}{(uh_0)^{a}}\,du
        \;\longrightarrow\; \mathcal{T}_m\otimes\Omega(0),
    \end{align*}
    by dominated convergence and continuity of $\Omega$ at $0$.
    
    To verify the Lindeberg condition of Lemma \ref{lemma:LF_CLT}, set
    $Y_{n,i}=(nh_0^{a+1})^{-1/2}Z_i'\zeta_i K(h_i/h_0)$, so that
    $\sum_{i=1}^n \operatorname{Cov}(Y_{n,i})\to\mathcal{T}_m\otimes\Omega(0)$ by
    the display above. Boundedness of $\normm{\nu_i}_2$ and of $X_i$, $\beta(X_i)$,
    $\delta_0$ and $g_y$ on $[0,h_0]$ yields
    \begin{align*}
        \normm{Z_i'\zeta_i K(h_i/h_0)}_2
        = \normm{\left(\nu_i\otimes W_i^{*\prime}\zeta_i\right)K(h_i/h_0)}_2
        \lesssim \left|K(h_i/h_0)\right|\normm{\zeta_i}_2
        \lesssim \left|K(h_i/h_0)\right|\left(\normm{Y_i}_2+1\right)
        \lesssim \left|K(h_i/h_0)\right|\left(\normm{\epsilon_i}_2+1\right).
    \end{align*}
    Since $h_i$ is a function of $X_i$, Condition \ref{cond:Regularity} and
    iterated expectations give
    $\E\left[|K(h_i/h_0)|^{2+\delta}(\normm{\epsilon_i}_2+1)^{2+\delta}\right]
    \lesssim \E|K(h_i/h_0)|^{2+\delta}$, so by (\ref{eq:kernel_moment}) the
    Lyapunov ratio satisfies
    \begin{align*}
        \sum_{i=1}^n\E\normm{Y_{n,i}}_2^{2+\delta}
        \;\lesssim\; \frac{n}{(nh_0^{a+1})^{1+\delta/2}}\,\E\left|K(h_i/h_0)\right|^{2+\delta}
        \;\lesssim\; \frac{nh_0^{a+1}}{(nh_0^{a+1})^{1+\delta/2}}
        = \frac{1}{(nh_0^{a+1})^{\delta/2}}\;\to\;0 ,
    \end{align*}
    using $nh_0^{a+1}\to\infty$. A Lyapunov condition implies the Lindeberg
    condition, so Lemma \ref{lemma:LF_CLT} applies and
    \begin{align}\label{eq:delta_Vn_convergence}
        \sqrt{nh_0^{a+1}}\,V_n
        = \frac{1}{\sqrt{nh_0^{a+1}}}\sum_{i=1}^n Z_i'\zeta_i K(h_i/h_0)
        \;\to_d\; N\!\left(0,\ \mathcal{T}_m\otimes \Omega(0)\right).
    \end{align}
    In particular $V_n=O_p\!\left((nh_0^{a+1})^{-1/2}\right)$.

    \textbf{Step 7 (Conclusion).}
    Write $\mathcal{M}:=\mathcal{S}_m\otimes\mathcal{A}(0)$. By Step 4,
    $S_n^{-1}-\mathcal{M}^{-1}=o_p(1)$, hence
    $\left(S_n^{-1}-\mathcal{M}^{-1}\right)V_n = o_p(1)\,O_p\!\left((nh_0^{a+1})^{-1/2}\right)
    = o_p\!\left((nh_0^{a+1})^{-1/2}\right)$. Combining this with
    (\ref{eq:theta_decomposition}) and (\ref{eq:delta_Lambda_convergence}),
    \begin{align*}
        \hat{\theta}-\theta_0
        = S_n^{-1}\Lambda_n + S_n^{-1}V_n
        = O_p\!\left(h_0^{m+1}\right) + o_p\!\left((nh_0^{a+1})^{-1/2}\right)
        + \left(\mathcal{S}_m\otimes \mathcal{A}(0)\right)^{-1} V_n .
    \end{align*}
    Applying (\ref{eq:delta_estimation_equivalence}) and the mixed-product rule
    $\left(e_1'\otimes I_{p(T-1)}\right)\left(\mathcal{S}_m^{-1}\otimes\mathcal{A}(0)^{-1}\right)
    = e_1'\mathcal{S}_m^{-1}\otimes\mathcal{A}(0)^{-1}$,
    \begin{align*}
        \hat{\delta}_0-\delta_0
        = O_p\!\left(h_0^{m+1}\right) + o_p\!\left((nh_0^{a+1})^{-1/2}\right)
        + \left(e_1'\mathcal{S}_m^{-1}\otimes \mathcal{A}(0)^{-1}\right) V_n .
    \end{align*}
    Finally, by (\ref{eq:delta_Vn_convergence}) and the continuous mapping theorem,
    \begin{align*}
        \sqrt{nh_0^{a+1}}\left(e_1'\mathcal{S}_m^{-1}\otimes \mathcal{A}(0)^{-1}\right) V_n
        \;\to_d\; N\!\left(0,\ \Sigma_\delta\right),
    \end{align*}
    with, using the symmetry of $\mathcal{S}_m$ and $\mathcal{A}(0)$ and the
    mixed-product rule once more,
    \begin{align*}
        \Sigma_\delta
        &= \left(e_1'\mathcal{S}_m^{-1}\otimes \mathcal{A}(0)^{-1}\right)
        \left(\mathcal{T}_m\otimes \Omega(0)\right)
        \left(\mathcal{S}_m^{-1}e_1\otimes \mathcal{A}(0)^{-1}\right)\\
        &= \left(e_1'\mathcal{S}_m^{-1}\mathcal{T}_m\mathcal{S}_m^{-1}e_1\right)
        \otimes\left(\mathcal{A}(0)^{-1}\Omega(0)\mathcal{A}(0)^{-1}\right)
        = \left(e_1'\mathcal{S}_m^{-1}\mathcal{T}_m\mathcal{S}_m^{-1}e_1\right)
        \mathcal{A}(0)^{-1}\Omega(0)\mathcal{A}(0)^{-1},
    \end{align*}
    the last equality holding because the first factor is a scalar. This completes
    the proof.

\end{proof}

\begin{proof}[Proof of Lemma~\ref{lemma:beta0_positive}]

\textbf{Step 1 (Decomposition).}
Since $T=p$ and $\det(X_i)=d_i$, we have
$(X_i'X_i)^{-1}X_i'=X_i^{-1}=d_i^{-1}X_i^{*}$ whenever $d_i\ne 0$, and hence
$(X_i'X_i)^{-1}X_i'W_i=d_i^{-1}W_i^{*}$. Substituting
$Y_i-W_i\hat\delta = X_i\beta(X_i)+\epsilon_i-W_i\Delta\delta$ gives, on
$\{h_i\ge h_0\}$,
\begin{align*}
    \hat{\beta}(X_i)
    &= (X_i'X_i)^{-1}X_i'\left(X_i\beta(X_i)+\epsilon_i-W_i\Delta\delta\right)\\
    &= \beta(X_i) + v_i - (X_i'X_i)^{-1}X_i'W_i\Delta\delta
     = \beta(X_i) + v_i - \frac{1}{d_i}W_i^{*}\Delta\delta .
\end{align*}
Consequently, since $\Delta\delta$ does not depend on $i$,
\begin{align}\label{eq:beta0_decomposition}
    \E_n \hat{\beta}(X_i)\I[h_i\geq h_0] - \E \beta(X_i)\I[h_i\geq h_0]
    = \E_n v_i\I[h_i\geq h_0]
      - \left(\E_n \frac{1}{d_i}W_i^{*}\I[h_i\geq h_0]\right)\Delta\delta
      + \left(\E_n-\E\right)\beta(X_i)\I[h_i\geq h_0].
\end{align}
By Condition~\ref{cond:Regularity}.2 the summands
$\beta(X_i)\I[h_i\ge h_0]$ are i.i.d.\ and uniformly bounded, so their variance
is bounded by a constant that does not depend on $h_0$; Chebyshev's inequality
then yields
\begin{align}\label{eq:beta0_root_n}
    \left(\E_n-\E\right)\beta(X_i)\I[h_i\geq h_0]=O_p\!\left(n^{-1/2}\right).
\end{align}
It therefore remains to show that
$\E_n d_i^{-1}W_i^{*}\I[h_i\geq h_0]\to_p\Xi_0$, which we do in Steps 2--4.

\textbf{Step 2 (Existence of $\Xi_0$ and convergence of the mean).}
Write $\psi(d):=g_W(d)q_D(d)$ for $|d|\le\bar h$. By
Condition~\ref{cond:Smoothness}.2, $\psi\in C^{1}[-\bar h,\bar h]$, so
$L:=\sup_{|d|\le\bar h}\normm{\psi'(d)}_F<\infty$. By the law of iterated
expectations and $h_i=|d_i|$,
\begin{align*}
    \E\!\left[\frac{1}{d_i}W_i^{*}\I[h_i\ge h]\right]
    = \E\!\left[\frac{1}{d_i}g_W(d_i)\I[\,|d_i|\ge h\,]\right]
    = \int_{h\le|d|\le\bar h}\frac{g_W(d)f_D(d)}{d}\,dd \;+\; R,
\end{align*}
where $R:=\E\left[d_i^{-1}g_W(d_i)\I[\,|d_i|>\bar h\,]\right]$ satisfies
$\normm{R}_F\le C/\bar h$ and does not depend on $h$. Substituting $d\mapsto-d$
on the negative half-line and using $f_D(d)=|d|^{a}q_D(d)$,
\begin{align}\label{eq:pv_pairing}
    \int_{h\le|d|\le\bar h}\frac{g_W(d)f_D(d)}{d}\,dd
    = \int_{h}^{\bar h}\frac{g_W(d)f_D(d)-g_W(-d)f_D(-d)}{d}\,dd
    = \int_{h}^{\bar h} d^{\,a-1}\left(\psi(d)-\psi(-d)\right)dd .
\end{align}
Because $\psi$ is matrix valued, we use the integral form of the mean value
theorem rather than a one-point expansion:
\begin{align}\label{eq:mvt}
    \psi(d)-\psi(-d)=\int_{-d}^{d}\psi'(t)\,dt,
    \qquad\text{so}\qquad
    \normm{\psi(d)-\psi(-d)}_F\le 2Ld \quad\text{for } d\in(0,\bar h].
\end{align}
Hence the integrand in (\ref{eq:pv_pairing}) is bounded in norm by $2Ld^{a}$,
which is integrable on $(0,\bar h]$ because $a>-1$. Dominated convergence
therefore shows that the limit (\ref{eq:Xi0_def}) exists, with
\begin{align}\label{eq:mean_rate}
    \normm{\Xi_0}_F \le \frac{C}{\bar h} + \frac{2L\,\bar h^{\,a+1}}{a+1}<\infty ,\quad     \normm{\E\!\left[\frac{1}{d_i}W_i^{*}\I[h_i\ge h_0]\right]-\Xi_0}_F
    \le \int_{0}^{h_0} 2L\,d^{\,a}\,dd = \frac{2L}{a+1}h_0^{\,a+1}= O\!\left(h_0^{\,a+1}\right)=o(1).
\end{align}

\textbf{Step 3 (Variance).}
By i.i.d.\ sampling, $\normm{W_i^{*}}_F\le C$ and $h_i=|d_i|$,
\begin{align*}
    \E\normm{\E_n \frac{1}{d_i}W_i^{*}\I[h_i\ge h_0]
      -\E \frac{1}{d_i}W_i^{*}\I[h_i\ge h_0]}_F^2
    \le \frac{1}{n}\,\E\normm{\frac{1}{d_i}W_i^{*}\I[h_i\ge h_0]}_F^2
    \lesssim \frac{1}{n}\,\E\!\left[\frac{1}{h_i^{2}}\I[h_i\ge h_0]\right].
\end{align*}
Splitting at $\bar h$ and using $f(h)\lesssim h^{a}$ on $(0,\bar h]$,
\begin{align*}
    \E\!\left[\frac{1}{h_i^{2}}\I[h_i\ge h_0]\right]
    \le \frac{1}{\bar h^{2}} + \int_{h_0}^{\bar h}\frac{f(h)}{h^{2}}\,dh
    \lesssim 1 + \int_{h_0}^{\bar h}h^{\,a-2}\,dh
    \lesssim
    \begin{cases}
        h_0^{\,a-1}, & a<1,\\[2pt]
        \log(1/h_0), & a=1,\\[2pt]
        1, & a>1 .
    \end{cases}
\end{align*}
In each case the bound divided by $n$ is $o(1)$. For $a<1$ this is because
$n^{-1}h_0^{\,a-1}=(nh_0^{1-a})^{-1}\to0$ by
Condition~\ref{cond:Kernel}.3. For $a=1$, $nh_0^{a+1}\to\infty$ gives
$nh_0^{2}\to\infty$, hence $h_0\gtrsim n^{-1/2}$ and
$n^{-1}\log(1/h_0)=O(n^{-1}\log n)\to0$. For $a>1$ the bound is $O(n^{-1})$.
Therefore
\begin{align}\label{eq:var_negligible}
    \E_n \frac{1}{d_i}W_i^{*}\I[h_i\ge h_0]
    = \E \frac{1}{d_i}W_i^{*}\I[h_i\ge h_0] + o_p(1).
\end{align}

\textbf{Step 4 (Conclusion).}
Combining (\ref{eq:mean_rate}) and (\ref{eq:var_negligible}),
\begin{align*}
    \E_n \frac{1}{d_i}W_i^{*}\I[h_i\ge h_0] = \Xi_0 + o_p(1),
\end{align*}
so that, $\Delta\delta$ being common to all $i$,
\begin{align*}
    \normm{\left(\E_n \frac{1}{d_i}W_i^{*}\I[h_i\ge h_0]\right)\Delta\delta
    -\Xi_0\Delta\delta}_2
    \le \normm{\E_n \frac{1}{d_i}W_i^{*}\I[h_i\ge h_0]-\Xi_0}_F
        \normm{\Delta\delta}_2
    = o_p\!\left(\normm{\Delta\delta}_2\right).
\end{align*}
Substituting this and (\ref{eq:beta0_root_n}) into
(\ref{eq:beta0_decomposition}) gives
\begin{align*}
    \E_n \hat{\beta}(X_i)\I[h_i\geq h_0] - \E \beta(X_i)\I[h_i\geq h_0]
    = \E_n v_i\I[h_i\geq h_0] - \Xi_0\Delta\delta
      + o_p\!\left(\normm{\Delta\delta}_2\right) + O_p\!\left(n^{-1/2}\right),
\end{align*}
which is the assertion of the lemma.

\end{proof}

\begin{proof}[Proof of Lemma~\ref{lemma:beta0_negative}]
    Throughout, let $\nu(t) := (1,t,\dots,t^m)' \in\R^{m+1},\ \mathcal{H}(u):=\nu(u)\nu(u)', \ \varsigma_i = \beta(X_i)-g_{\beta}(h_i)$

\textbf{Step 1 (Reduction).}
Decompose
\begin{align}\label{eq:beta0neg_split}
    \E_n \hat{g}_{\beta}(h_i)\I[h_i< h_0] - \E g_{\beta}(h_i)\I[h_i< h_0]
    = \E_n\!\left[\left(\hat{g}_{\beta}(h_i)-g_{\beta}(h_i)\right)\I[h_i< h_0]\right]
      + \left(\E_n-\E\right)g_{\beta}(h_i)\I[h_i< h_0].
\end{align}
The summands of the second term are i.i.d., bounded by
Condition~\ref{cond:Regularity}.2, and supported on $\{h_i<h_0\}$, so
\begin{align*}
    \Var\!\left(\left(\E_n-\E\right)g_{\beta}(h_i)\I[h_i<h_0]\right)
    \le \frac{1}{n}\,\E\!\left[\normm{g_\beta(h_i)}_2^2\I[h_i<h_0]\right]
    \lesssim \frac{P(h_i<h_0)}{n}\asymp \frac{h_0^{a+1}}{n}=o\!\left(n^{-1}\right),
\end{align*}
whence, by Chebyshev's inequality,
\begin{align}\label{eq:beta0neg_sampling}
    \left(\E_n-\E\right)g_{\beta}(h_i)\I[h_i< h_0]=o_p\!\left(n^{-1/2}\right).
\end{align}
It remains to analyse the first term of (\ref{eq:beta0neg_split}).

\textbf{Step 2 (Matrix form of the local polynomial and rescaling).}
Collect the Step-2 coefficients into the $p\times(m+1)$ matrix
$\hat\gamma=(\hat\gamma_0,\dots,\hat\gamma_m)$, so that
$\hat g_\beta(h)=\hat\gamma\,\nu(h-h_0)$ and, from the FOC of
(\ref{eq:g_estimation}),
\begin{align*}
    \hat{\gamma} =\left(\sum_{i=1}^n\hat{\beta}(X_i)\,\nu(h_i-h_0)'
        K\!\left(\tfrac{h_i-h_0}{b}\right)\right)
        \left(\sum_{i=1}^n \nu(h_i-h_0)\nu(h_i-h_0)'
        K\!\left(\tfrac{h_i-h_0}{b}\right)\right)^{-1}.
\end{align*}
Since $g_\beta\in C^{m+1}[0,\bar h]$, a Taylor expansion \emph{at $h_0$} gives
\begin{align}\label{eq:taylor_gbeta}
    g_{\beta}(h_i) = \Gamma_0\,\nu(h_i-h_0) + R_{\beta}(h_i),
    \qquad
    \Gamma_0:=\left[\frac{g_{\beta}(h_0)}{0!},\dots,\frac{g_{\beta}^{(m)}(h_0)}{m!}\right]
    \in\R^{p\times(m+1)},
\end{align}
with $\normm{R_{\beta}(h_i)}_2\lesssim |h_i-h_0|^{m+1}$. (Note the expansion is
centred at $h_0$, not at $0$; on both regions used below,
$h_i<h_0$ and $h_i\in[h_0,(1+c)h_0]$, one has $|h_i-h_0|\lesssim h_0$.)
Then
\begin{align*}
    \hat{\beta}(X_i)
    = \beta(X_i) + v_i - d_i^{-1}W_i^{*}\Delta \delta
    = \Gamma_0\,\nu(h_i-h_0) + R_{\beta}(h_i) + \varsigma_i + v_i
      - d_i^{-1}W_i^{*}\Delta \delta ,
\end{align*}
so that, subtracting (\ref{eq:taylor_gbeta}), the $\Gamma_0$ terms cancel in
$\hat\gamma-\Gamma_0$ and
\begin{align*}
    \hat{g}_\beta(h_i) - g_\beta(h_i)
    = \left(\hat\gamma-\Gamma_0\right)\nu(h_i-h_0)-R_\beta(h_i)
    = A_n B_n^{-1}\nu(h_i-h_0)-R_{\beta}(h_i),
\end{align*}
with
$A_n=\sum_i\left(R_{\beta}(h_i)+\varsigma_i+v_i-d_i^{-1}W_i^{*}\Delta\delta\right)
\nu(h_i-h_0)'K(\tfrac{h_i-h_0}{b})$ and
$B_n=\sum_i\nu(h_i-h_0)\nu(h_i-h_0)'K(\tfrac{h_i-h_0}{b})$.
 
Let $D(b)=\mathrm{diag}(1,b,\dots,b^m)$, so that
$D(b)^{-1}\nu(h_i-h_0)=\nu\!\left(\tfrac{h_i-h_0}{b}\right)$, and write
$K_b(\cdot)=K(\cdot)/b$. Inserting $D(b)^{-1}D(b)$ twice, which leaves
$A_nB_n^{-1}\nu(h_i-h_0)$ unchanged, and noting that the factor $1/b$ in $K_b$
cancels between $A_n$ and $B_n^{-1}$,
\begin{align}\label{eq:ghat_expansion}
    \hat{g}_\beta(h_i) - g_\beta(h_i)
    = \left(\Lambda_n^R +\Lambda_n^{\varsigma} + \Lambda_n^{v} + \Lambda_n^{\delta}\right)
      Q_n^{-1}\,\nu\!\left(\tfrac{h_i-h_0}{b}\right)\frac{1}{h_0^{a+1}}
      -R_{\beta}(h_i),
\end{align}
where, with $u_i:=\tfrac{h_i-h_0}{b}$,
\begin{alignat*}{2}
    \Lambda_n^R &= \frac{h_0}{n}\sum_{i=1}^n R_{\beta}(h_i)\,\nu(u_i)'K_b(u_i),
    &\qquad
    \Lambda_n^{\varsigma} &= \frac{h_0}{n}\sum_{i=1}^n \varsigma_i\,\nu(u_i)'K_b(u_i),\\
    \Lambda_n^{v} &= \frac{h_0}{n}\sum_{i=1}^n v_i\,\nu(u_i)'K_b(u_i),
    &\qquad
    \Lambda_n^{\delta} &= -\frac{h_0}{n}\sum_{i=1}^n d_i^{-1}W_i^{*}\Delta \delta\,
        \nu(u_i)'K_b(u_i),\\
     Q_n &= \frac{1}{nh_0^a}\sum_{i=1}^n \nu(u_i)\nu(u_i)'K_b(u_i).
\end{alignat*}
 Indeed, with
these normalisations
$\frac{h_0}{n}A_n^{D}\left(\frac{1}{nh_0^{a}}B_n^{D}\right)^{-1}
=h_0^{a+1}A_n^{D}(B_n^{D})^{-1}$, where the superscript $D$ denotes the
$D(b)^{-1}$-rescaled matrices, which is the origin of the factor $h_0^{-(a+1)}$
in (\ref{eq:ghat_expansion}).
 
Averaging (\ref{eq:ghat_expansion}) over $\{h_i<h_0\}$, and using that
$\Lambda_n^{\bullet}$ and $Q_n$ do not depend on $i$,
\begin{align}\label{eq:beta0neg_master}
    \E_n\!\left[\left(\hat{g}_\beta(h_i) - g_\beta(h_i)\right)\I[h_i<h_0]\right]
    = \left(\Lambda_n^R +\Lambda_n^{\varsigma} + \Lambda_n^{v} + \Lambda_n^{\delta}\right)
      Q_n^{-1}\hat{\tau}_n - R_n,
\end{align}
where
\begin{alignat*}{2}
    \hat{\tau}_n &:= \E_n\!\left[\nu(u_i)\I[h_i<h_0]\right]/h_0^{a+1},
    &\qquad
    R_n &:= \E_n\!\left[R_{\beta}(h_i)\I[h_i<h_0]\right].
\end{alignat*}
Observe that $\hat\tau_n$ is built from the stayers ($u_i\in[-1/c,0)$, where the
kernel does not appear), whereas the $\Lambda_n^{\bullet}$ are built from the
movers with $h_i\in[h_0,(1+c)h_0]$, where the kernel does appear.

\textbf{Step 3 (Limit of $Q_n$).}
Note $\mathcal{H}(u)=\nu(u)\nu(u)'=(u^{k+l-2})_{1\le k,l\le m+1}$. Substituting
$h=h_0+bu$,
\begin{align*}
    \E Q_n
    = \frac{1}{h_0^a}\,\E\!\left[\mathcal{H}(u_i)K_b(u_i)\right]
    = \frac{1}{h_0^{a}}\int_{0}^{1}\mathcal{H}(u)K(u)f(h_0+bu)\,du
    = \alpha\int_{0}^{1}\mathcal{H}(u)K(u)(1+cu)^{a}\,
      \frac{f(h_0+bu)}{\alpha(h_0+bu)^{a}}\,du,
\end{align*}
and since $f(h)/(\alpha h^{a})\to1$ as $h\downarrow0$, the dominated convergence theorem gives
\begin{align}\label{eq:Qm_def}
    \E Q_n \;\longrightarrow\;
    \mathcal{Q}_m:=\alpha\int_{0}^{1}\mathcal{H}(u)K(u)(1+cu)^{a}\,du .
\end{align}
For the $(k,l)$ entry, using $K_b^2=K^2/b^2$, $dh=b\,du$ and $b=ch_0$,
\begin{align*}
    \Var(Q_{n,kl})
    \le \frac{1}{nh_0^{2a}}\,\E\!\left[u_i^{2(k+l-2)}K_b^2(u_i)\right]
    = \frac{1}{nh_0^{2a}b}\int_{0}^{1}u^{2(k+l-2)}K^2(u)f(h_0+bu)\,du
    \lesssim \frac{1}{nh_0^{a+1}}\;\to\;0,
\end{align*}
by Condition~\ref{cond:Kernel}.3. Hence $Q_n\to_p\mathcal{Q}_m$. The limit is
invertible: for $\lambda\in\R^{m+1}$, $\lambda\ne0$,
$\lambda'\mathcal{Q}_m\lambda=\alpha\int_0^1\left(\nu(u)'\lambda\right)^2K(u)(1+cu)^{a}du>0$,
since a non-zero polynomial has finitely many roots while $K$ is strictly
positive on a set of positive measure. Therefore
$Q_n^{-1}\to_p\mathcal{Q}_m^{-1}$.

\textbf{Step 4 (Limit of $\hat\tau_n$).}
For the $k$-th entry, $k=0,\dots,m$, substituting $h=h_0u$ and using
$\left(\tfrac{u-1}{c}\right)^{k}=(-1/c)^{k}(1-u)^{k}$,
\begin{align*}
    \E\hat\tau_{n,k}
    = \frac{1}{h_0^{a+1}}\int_{0}^{h_0}\left(\tfrac{h-h_0}{b}\right)^{k}f(h)\,dh
    = \frac{1}{h_0^{a}}\int_{0}^{1}\left(\tfrac{u-1}{c}\right)^{k}f(h_0u)\,du
    = \alpha\left(-\tfrac{1}{c}\right)^{k}\int_{0}^{1}(1-u)^{k}u^{a}\,
      \frac{f(h_0u)}{\alpha (h_0u)^{a}}\,du,
\end{align*}
so that
\begin{align}\label{eq:taum_def}
    \E\hat\tau_{n,k}\;\longrightarrow\;
    \tau_{m,k}:=\alpha\left(-\tfrac{1}{c}\right)^{k}B(k+1,a+1),
    \qquad \tau_m=(\tau_{m,0},\dots,\tau_{m,m})',
\end{align}
where $B(\cdot,\cdot)$ is the Beta function and the integral is finite because
$a>-1$. Moreover
\begin{align*}
    \Var(\hat\tau_{n,k})
    \le \frac{1}{nh_0^{2(a+1)}}\int_{0}^{h_0}\left(\tfrac{h-h_0}{b}\right)^{2k}f(h)\,dh
    = \frac{1}{nh_0^{2a+1}}\int_{0}^{1}\left(\tfrac{u-1}{c}\right)^{2k}f(h_0u)\,du
    \lesssim \frac{1}{nh_0^{a+1}}\to0 .
\end{align*}
Hence $\hat\tau_n\to_p\tau_m$ and, with $\omega_m:=\mathcal{Q}_m^{-1}\tau_m$,
\begin{align}\label{eq:omega_conv}
    Q_n^{-1}\hat\tau_n = \omega_m + o_p(1).
\end{align}

\textbf{Step 5 (The smoothing bias: $R_n$ and $\Lambda_n^R$).}
Since $|h_i-h_0|\le h_0$ on $\{h_i<h_0\}$,
\begin{align*}
    \E\normm{R_n}_2
    \le \E\!\left[\normm{R_{\beta}(h_i)}_2\I[h_i< h_0]\right]
    \lesssim h_0^{m+1}P(h_i<h_0)\lesssim h_0^{m+a+2}.
\end{align*}
Similarly, on the kernel support $h_i\in[h_0,(1+c)h_0]$ we have
$|h_i-h_0|\le ch_0$, so
\begin{align*}
    \E\normm{\Lambda_n^R}_F
    \lesssim h_0\,\E\!\left[\normm{R_{\beta}(h_i)}_2 K_b(u_i)\right]
    \lesssim h_0^{m+2}\,\E\!\left[K_b(u_i)\right]
    = h_0^{m+2}\int_{0}^{1}K(u)f(h_0+bu)\,du
    \lesssim h_0^{m+a+2}.
\end{align*}
By Markov's inequality, $R_n=O_p(h_0^{m+a+2})$ and
$\Lambda_n^R=O_p(h_0^{m+a+2})$; combined with (\ref{eq:omega_conv}), which gives
$Q_n^{-1}\hat\tau_n=O_p(1)$,
\begin{align}\label{eq:bias_terms}
    \Lambda_n^{R}Q_n^{-1}\hat\tau_n-R_n=O_p\!\left(h_0^{m+a+2}\right).
\end{align}

\textbf{Step 6 (The first-step term $\Lambda_n^{\delta}$).}
The $l$-th column of $\Lambda_n^{\delta}$ equals $\Psi_n^{(l)}\Delta\delta$ with
\begin{align*}
    \Psi_n^{(l)} := -\frac{h_0}{n}\sum_{i=1}^n u_i^{\,l-1}K_b(u_i)\,d_i^{-1}W_i^{*}
    \in\R^{p\times p(T-1)},\qquad l=1,\dots,m+1 .
\end{align*}
\emph{Variance.} Using $\normm{W_i^{*}}_F\le C$, $|u_i|\le1$, $h_0^2K_b^2\lesssim K^2$,
$h_i=|d_i|\asymp h_0$ on the kernel support, and $f(h)\asymp h^{a}$,
\begin{align*}
    \E\normm{\Psi_n^{(l)}-\E\Psi_n^{(l)}}_F^2
    \lesssim \frac{h_0^2}{n}\,\E\!\left[K_b^2(u_i)\,d_i^{-2}\right]
    \lesssim \frac{1}{n}\int_{0}^{1}K^2(u)(h_0+bu)^{-2}f(h_0+bu)\,b\,du
    \lesssim \frac{h_0^{a-1}}{n},
\end{align*}
which tends to $0$ by $nh_0^{1-a}\to\infty$.
 
\emph{Mean.} Write $\kappa_l(d):=\left(\tfrac{|d|-h_0}{b}\right)^{l-1}
K\!\left(\tfrac{|d|-h_0}{b}\right)$, which is a bounded \emph{even} function of
$d$ supported on $\{h_0\le|d|\le(1+c)h_0\}$, and $\psi(d):=g_W(d)q_D(d)$, which
is $C^1$ on $[-\bar h,\bar h]$ with $L:=\sup_{|d|\le\bar h}\normm{\psi'(d)}_F<\infty$
by Condition~\ref{cond:Smoothness}.2. Since $h_0K_b=K/c$ and
$f_D(d)=|d|^{a}q_D(d)$, conditioning on $d_i$ and pairing $\pm d$ gives
\begin{align*}
    \E\Psi_n^{(l)}
    = -\frac{1}{c}\int_{h_0}^{(1+c)h_0}
      \frac{\kappa_l(d)g_W(d)f_D(d)-\kappa_l(-d)g_W(-d)f_D(-d)}{d}\,dd
    = -\frac{1}{c}\int_{h_0}^{(1+c)h_0}\kappa_l(d)
      \left(\psi(d)-\psi(-d)\right)d^{\,a-1}\,dd,
\end{align*}
where the evenness of $\kappa_l$ is what allows it to be factored out of the
difference. By the integral form of the mean value theorem,
$\normm{\psi(d)-\psi(-d)}_F=\normm{\int_{-d}^{d}\psi'(t)dt}_F\le 2Ld$, so
\begin{align*}
    \normm{\E\Psi_n^{(l)}}_F
    \lesssim \int_{h_0}^{(1+c)h_0}d^{\,a}\,dd
    = h_0^{a+1}\int_{1}^{1+c}u^{a}\,du\lesssim h_0^{a+1}.
\end{align*}
Combining the two displays, $\Psi_n^{(l)}=O(h_0^{a+1})+O_p(\sqrt{h_0^{a-1}/n})$.
The stochastic part need not be dominated by the mean, but it is harmless:
since $\Delta\delta=o_p(1)$ by Lemma~\ref{lemma:delta_convergence},
\begin{align}\label{eq:Lambda_delta}
    \Lambda_n^{\delta}Q_n^{-1}\hat\tau_n
    = O_p\!\left(h_0^{a+1}\right)\Delta\delta
      + o_p\!\left(\sqrt{h_0^{a-1}/n}\right).
\end{align}

\textbf{Step 7 (The stochastic terms $\Lambda_n^{\varsigma}$ and $\Lambda_n^{v}$).}
Both have mean zero: $\E[\varsigma_i\mid h_i]=0$ by the definition of
$g_\beta$, and $\E[v_i\mid X_i]=0$ by $\E[\epsilon_i\mid X_i]=0$. For any fixed
$\lambda\in\R^{m+1}$, substituting $h=h_0+bu$ and using $b=ch_0$,
\begin{align*}
    \Var\!\left(\Lambda_n^{\varsigma}\lambda\right)
    &= \frac{h_0^2}{n}\,
       \E\!\left[\left(\nu(u_i)'\lambda\right)^2K_b^2(u_i)\Sigma_\beta(h_i)\right]
     = \frac{h_0^2}{nb}\int_{0}^{1}\left(\nu(u)'\lambda\right)^2K^2(u)
       \Sigma_\beta(h_0+bu)f(h_0+bu)\,du
     = O\!\left(\frac{h_0^{a+1}}{n}\right),\\
    \Var\!\left(\Lambda_n^{v}\lambda\right)
    &= \frac{h_0^2}{n}\,
       \E\!\left[\left(\nu(u_i)'\lambda\right)^2K_b^2(u_i)h_i^{-2}\Sigma_v(h_i)\right]
     = \frac{h_0^2}{nb}\int_{0}^{1}\left(\nu(u)'\lambda\right)^2K^2(u)
       \frac{\Sigma_v(h_0+bu)f(h_0+bu)}{(h_0+bu)^{2}}\,du
     = O\!\left(\frac{h_0^{a-1}}{n}\right),
\end{align*}
using $\normm{\Sigma_\beta}_F\le C$, $\normm{\Sigma_v}_F\le C$ and
$f(h)\asymp h^{a}$. Taking $\lambda=\omega_m$ and, separately,
\begin{align}\label{eq:stoch_rates}
    \Lambda_n^{\varsigma}\omega_m = O_p\!\left(\sqrt{h_0^{a+1}/n}\right)
    = o_p\!\left(n^{-1/2}\right),
    \qquad
    \Lambda_n^{v}\omega_m = O_p\!\left(\sqrt{h_0^{a-1}/n}\right),
    \qquad
    \normm{\Lambda_n^{\varsigma}}_F+\normm{\Lambda_n^{v}}_F
    = O_p\!\left(\sqrt{h_0^{a-1}/n}\right),
\end{align}
the first equality because $h_0^{a+1}\to0$. Consequently, by
(\ref{eq:omega_conv}),
\begin{align}\label{eq:stoch_replace}
    \left(\Lambda_n^{\varsigma}+\Lambda_n^{v}\right)Q_n^{-1}\hat\tau_n
    = \left(\Lambda_n^{\varsigma}+\Lambda_n^{v}\right)\omega_m
      + \left(\Lambda_n^{\varsigma}+\Lambda_n^{v}\right)
        \left(Q_n^{-1}\hat\tau_n-\omega_m\right)
    = \left(\Lambda_n^{\varsigma}+\Lambda_n^{v}\right)\omega_m
      + o_p\!\left(\sqrt{h_0^{a-1}/n}\right).
\end{align}

\textbf{Step 8 (Collecting terms).}
Substituting (\ref{eq:bias_terms}), (\ref{eq:Lambda_delta}) and
(\ref{eq:stoch_replace}) into (\ref{eq:beta0neg_master}),
\begin{align*}
    \E_n\!\left[\left(\hat{g}_\beta(h_i)-g_\beta(h_i)\right)\I[h_i<h_0]\right]
    = \left(\Lambda_n^{\varsigma}+\Lambda_n^{v}\right)\omega_m
      + o_p\!\left(\sqrt{h_0^{a-1}/n}\right)
      + O_p\!\left(h_0^{a+1}\right)\Delta\delta
      + O_p\!\left(h_0^{m+a+2}\right),
\end{align*}
and by the first statement in (\ref{eq:stoch_rates}),
$\Lambda_n^{\varsigma}\omega_m=o_p(n^{-1/2})$. Adding
(\ref{eq:beta0neg_sampling}) yields
\begin{align*}
    \E_n \hat{g}_\beta(h_i)\I[h_i<h_0] - \E g_\beta(h_i)\I[h_i<h_0]
    = \Lambda_n^{v}\omega_m
      + o_p\!\left(\sqrt{h_0^{a-1}/n}\right)
      + O_p\!\left(h_0^{a+1}\right)\Delta\delta
      + O_p\!\left(h_0^{m+a+2}\right)
      + o_p\!\left(n^{-1/2}\right),
\end{align*}
which is the assertion of the lemma.

\end{proof}

\section{Proof of the main theorem}

\begin{proof}[Proof of Theorem~\ref{theorem:asymptotic}]
 
\textbf{Step 1 (Combining the lemmas).}
By the law of iterated expectations,
$\beta_0=\E\beta(X_i)=\E\left[\beta(X_i)\I[h_i\ge h_0]\right]
+\E\left[g_\beta(h_i)\I[h_i<h_0]\right]$, so that $\hat\beta_0-\beta_0$ is the
sum of the two left-hand sides of Lemmas~\ref{lemma:beta0_positive} and
\ref{lemma:beta0_negative}. Writing $\Delta\delta=\hat\delta_0-\delta_0$ and
$u_i=(h_i-h_0)/b$, and absorbing $O_p(h_0^{a+1})\Delta\delta$ into
$o_p(\normm{\Delta\delta}_2)$ (legitimate since $h_0^{a+1}\to0$), those two
lemmas give
\begin{align}\label{eq:beta0_master}
    \hat{\beta}_0 - \beta_0
    = \E_n v_i\I[h_i\geq h_0] + \Lambda_n^{v}\omega_m - \Xi_0\Delta\delta
      + o_p\!\left(\normm{\Delta\delta}_2\right)
      + O_p\!\left(n^{-1/2}\right)
      + o_p\!\left(\sqrt{h_0^{a-1}/n}\right)
      + O_p\!\left(h_0^{m+a+2}\right).
\end{align}
Since $h_0K_b(\cdot)=K(\cdot)/c$, the second term is itself a sample average:
\begin{align*}
    \Lambda_n^{v}\omega_m
    = \frac{h_0}{n}\sum_{i=1}^n v_i\,\nu(u_i)'K_b(u_i)\,\omega_m
    = \frac{1}{n}\sum_{i=1}^n v_i\,a_{+}(u_i),
    \qquad a_{+}(u):=\frac{1}{c}\nu(u)'K(u)\omega_m ,
\end{align*}
where $a_{+}$ is bounded and supported on $u\in[0,1]$, i.e.\ on
$h_i\in[h_0,(1+c)h_0]$. Hence, setting
\begin{align}\label{eq:psi_def}
    \psi_i := v_i\left(\I[h_i\geq h_0]+a_{+}(u_i)\right)\in\R^{p},
\end{align}
equation (\ref{eq:beta0_master}) reads
\begin{align}\label{eq:beta0_master2}
    \hat{\beta}_0 - \beta_0
    = \E_n\psi_i - \Xi_0\Delta\delta
      + o_p\!\left(\normm{\Delta\delta}_2\right)
      + O_p\!\left(n^{-1/2}\right)
      + o_p\!\left(\sqrt{h_0^{a-1}/n}\right)
      + O_p\!\left(h_0^{m+a+2}\right),
\end{align}
while Lemma~\ref{lemma:delta_convergence} gives
\begin{align}\label{eq:delta_master}
    \Delta\delta
    = O_p\!\left(h_0^{m+1}\right) + o_p\!\left((nh_0^{a+1})^{-1/2}\right)
      + \left(e_1'\mathcal{S}_m^{-1}\otimes \mathcal{A}(0)^{-1}\right)V_n .
\end{align}
Note that $\psi_i$ vanishes unless $h_i\ge h_0$, whereas every summand of $V_n$
carries the factor $K(h_i/h_0)$ and so vanishes unless $h_i\le h_0$.

\textbf{Step 2 (Variance of $\E_n\psi_i$).}
Since $\E[v_i\mid X_i]=0$ we have $\E\psi_i=0$, and since
$\E[v_iv_i'\mid h_i]=h_i^{-2}\Sigma_v(h_i)$,
\begin{align}\label{eq:psi_var}
    n\Var\!\left(\E_n\psi_i\right)
    = \E\!\left[h_i^{-2}\Sigma_v(h_i)
      \left(\I[h_i\ge h_0]+a_{+}(u_i)\right)^{2}\right]
    = \int_{h_0}^{\infty} h^{-2}\Sigma_v(h)
      \Big(1+a_{+}\big(\tfrac{h-h_0}{b}\big)\Big)^{2}f(h)\,dh .
\end{align}
Substitute $h=h_0t$ and write $f(h)=h^{a}\varrho(h)$ with
$\varrho(h)=f(h)/h^{a}\to\alpha$ as $h\downarrow0$, $\varrho$ bounded. Because
$h^{-2}dh=h_0^{-1}t^{-2}dt$ and $(h-h_0)/b=(t-1)/c$, splitting the integral at $h=\bar h$ gives
\begin{align}\label{eq:psi_var_scaled}
    n h_0^{1-a}\Var\!\left(\E_n\psi_i\right)
    =& \int_{1}^{\infty}t^{a-2}\,\Sigma_v(h_0t)\,\varrho(h_0t)
      \Big(1+a_{+}\big(\tfrac{t-1}{c}\big)\Big)^{2}dt + h_0^{1-a}\int_{\bar h}^{\infty} h^{-2}\Sigma_v(h)
      \Big(1+a_{+}\big(\tfrac{h-h_0}{b}\big)\Big)^{2}f(h)\,dh \notag\\
    =& \int_{1}^{\infty}t^{a-2}\,\Sigma_v(h_0t)\,\varrho(h_0t)
      \Big(1+a_{+}\big(\tfrac{t-1}{c}\big)\Big)^{2}dt + O(h_0^{1-a}).
\end{align}
For each fixed $t$ the integrand converges to
$\alpha t^{a-2}\Sigma_v(0)\big(1+a_{+}(\tfrac{t-1}{c})\big)^{2}$, by $h_0\to0$,
the continuity of $\Sigma_v$ at $0$ (Condition~\ref{cond:Regularity}.3) and
$\varrho(h)\to\alpha$; and it is dominated by $Ct^{a-2}$, which is integrable on
$[1,\infty)$ precisely when $a<1$. Dominated convergence therefore gives, for
$a<1$,
\begin{align}\label{eq:Upsilon_limit}
    n h_0^{1-a}\Var\!\left(\E_n\psi_i\right)\;\longrightarrow\;\Upsilon_0
    =\alpha\varphi_0\Sigma_v(0),
\end{align}
with $\varphi_0$ as in (\ref{eq:Upsilon_def}); splitting the integral at
$t=1+c$ and substituting $t=1+cu$ on $[1,1+c]$, where $a_{+}$ is supported, and
using $a_{+}\equiv0$ on $(1+c,\infty)$, yields the second expression for
$\varphi_0$.
 
For later use we also record a bound valid for every $a>-1$: since
$\normm{\Sigma_v}_F\le C$, $a_{+}$ is bounded and $f(h)\le C h^{a}$ for $h\in[0,\bar{h}]$,
(\ref{eq:psi_var}) gives
\begin{align}\label{eq:psi_var_bound}
    \Var\!\left(\E_n\psi_i\right)
    \lesssim \frac{1}{n}(\int_{h_0}^{\bar h}h^{a-2}dh+\int_{\bar h}^{\infty}\bar{h}^{-2}\ f(h)dh)
    \lesssim \frac{1}{n}\times
    \begin{cases}
        h_0^{a-1}, & a<1,\\
        \log(1/h_0), & a=1,\\
        1, & a>1 .
    \end{cases}
\end{align}

\textbf{Step 3 (Central limit theorem for $\E_n\psi_i$).}
Let $a<1$ and put $Y_{n,i}=\sqrt{h_0^{1-a}/n}\,\psi_i$, so that
$\sum_{i=1}^n\Cov(Y_{n,i})\to\Upsilon_0$ by (\ref{eq:Upsilon_limit}). Since
$\normm{v_i}_2\lesssim h_i^{-1}\normm{\epsilon_i}_2$ and
$\E[\normm{\epsilon_i}_2^{2+\delta}\mid X_i]\le C$, iterated expectations give
\begin{align*}
    \E\normm{\psi_i}_2^{2+\delta}
    \lesssim \E\!\left[h_i^{-(2+\delta)}\I[h_i\ge h_0]\right]
    \lesssim \int_{h_0}^{\bar h}h^{a-2-\delta}dh\lesssim h_0^{a-1-\delta},
\end{align*}
so that the Lyapunov ratio equals
\begin{align*}
    \sum_{i=1}^n\E\normm{Y_{n,i}}_2^{2+\delta}
    \lesssim \frac{n\,h_0^{(1-a)(1+\delta/2)}}{n^{1+\delta/2}}\,h_0^{a-1-\delta}
    = n^{-\delta/2}h_0^{-(1+a)\delta/2}
    = \left(nh_0^{1+a}\right)^{-\delta/2}\;\to\;0
\end{align*}
by Condition~\ref{cond:Kernel}.3. A Lyapunov condition implies the Lindeberg
condition, so Lemma~\ref{lemma:LF_CLT} yields
\begin{align}\label{eq:psi_CLT}
    \sqrt{nh_0^{1-a}}\,\E_n\psi_i\;\to_d\;N\!\left(0,\Upsilon_0\right),
    \qquad a<1 .
\end{align}

\textbf{Step 4 (Zero covariance and the joint CLT).}
The covariance of $\{\psi_i\}$ and $\{Z_i'\zeta_iK(h_i/h_0)\}$ is $0$. Indeed, by Step 1 the product
$\psi_i\left(Z_i'\zeta_iK(h_i/h_0)\right)'$ is zero almost surely for every $i$
(the two factors are supported on $\{h_i\ge h_0\}$ and $\{h_i\le h_0\}$
respectively). Hence
\begin{align}\label{eq:zero_cov}
    \Cov\!\left(\E_n\psi_i,\;V_n\right)=0\qquad\text{for every }n .
\end{align}
Applying the Cramér--Wold device to the stacked array
$\left(\lambda_1'\psi_i,\;\lambda_2'Z_i'\zeta_iK(h_i/h_0)\right)$ — the
Lyapunov condition holds for each coordinate by Step 3 and by Step 6 of the
proof of Lemma~\ref{lemma:delta_convergence} — and using (\ref{eq:zero_cov}),
the two normalised sums converge jointly to independent Gaussian limits:
\begin{align}\label{eq:joint_CLT}
    \begin{pmatrix}
        \sqrt{nh_0^{1+a}}\,\Delta\delta\\[2pt]
        \sqrt{nh_0^{1-a}}\,\E_n\psi_i
    \end{pmatrix}
    \;\to_d\;
    N\!\left(0,\begin{pmatrix}\phi_1\Lambda_0&0\\0&\Upsilon_0\end{pmatrix}\right),
    \qquad a<1,
\end{align}
where the first block uses (\ref{eq:delta_master}) and holds for any $a>-1$, the undersmoothing condition
$nh_0^{a+2m+3}\to0$ (which makes $\sqrt{nh_0^{a+1}}O_p(h_0^{m+1})=o_p(1)$) and
Lemma~\ref{lemma:delta_convergence}.

\textbf{Step 5 (Case $a>0$).}
Multiply (\ref{eq:beta0_master2}) by $\sqrt{nh_0^{1+a}}$ and treat the terms in
turn. By (\ref{eq:psi_var_bound}) and Chebyshev,
\begin{align*}
    \sqrt{nh_0^{1+a}}\,\E_n\psi_i
    = O_p\!\left(\sqrt{h_0^{1+a}\left(h_0^{a-1}\vee\log(1/h_0)\vee1\right)}\right)
    = o_p(1),
\end{align*}
since the bracket equals $h_0^{2a}$ for $a<1$, $h_0^{2}\log(1/h_0)$ for $a=1$
and $h_0^{1+a}$ for $a>1$, each tending to $0$ when $a>0$. Next, Lemma~\ref{lemma:delta_convergence} gives
$\sqrt{nh_0^{1+a}}\normm{\Delta\delta}_2=O_p(1)$, so
$\sqrt{nh_0^{1+a}}\,o_p(\normm{\Delta\delta}_2)=o_p(1)$. Furthermore
\begin{align*}
    \sqrt{nh_0^{1+a}}\,O_p\!\left(n^{-1/2}\right)&=O_p\!\left(h_0^{(1+a)/2}\right)=o_p(1),\\
    \sqrt{nh_0^{1+a}}\,o_p\!\left(\sqrt{h_0^{a-1}/n}\right)&=o_p\!\left(h_0^{a}\right)=o_p(1),\\
    \sqrt{nh_0^{1+a}}\,O_p\!\left(h_0^{m+a+2}\right)
    &=O_p\!\left(\sqrt{nh_0^{a+2m+3}}\;h_0^{a+1}\right)=o_p(1),
\end{align*}
Hence $\sqrt{nh_0^{1+a}}(\hat\beta_0-\beta_0)
=-\Xi_0\sqrt{nh_0^{1+a}}\Delta\delta+o_p(1)$.
We have asymptotic normality for $\hat{\delta}_0-\delta_0$ by Step 4. 
The result for $\hat{\beta}_0-\beta_0$ and the joint statement follow from the result for $\hat{\delta}_0-\delta_0$,
 the continuous mapping theorem, and Slutsky's lemma.

\textbf{Step 6 (Case $a=0$).}
Here $\sqrt{nh_0^{1+a}}=\sqrt{nh_0^{1-a}}=\sqrt{nh_0}$, and the two leading
terms of (\ref{eq:beta0_master2}) are of the same order. The remainder terms
are handled as in Step 4 and 5 with $a=0$. Therefore
\begin{align*}
    \sqrt{nh_0}\begin{pmatrix}\hat{\delta}_0-\delta_0\\ \hat{\beta}_0-\beta_0\end{pmatrix}
    = \begin{pmatrix}I_{p(T-1)} & 0\\ -\Xi_0 & I_p\end{pmatrix}
      \sqrt{nh_0}\begin{pmatrix}\Delta\delta\\ \E_n\psi_i\end{pmatrix}
      + o_p(1),
\end{align*}
and (\ref{eq:joint_CLT}) with $a=0$ gives, by the continuous mapping theorem,
\begin{align*}
    \begin{pmatrix}I&0\\-\Xi_0&I\end{pmatrix}
    \begin{pmatrix}\phi_1\Lambda_0&0\\0&\Upsilon_0\end{pmatrix}
    \begin{pmatrix}I&-\Xi_0'\\0&I\end{pmatrix}
    =\begin{pmatrix}
        \phi_1\Lambda_0 & -\phi_1\Lambda_0\Xi_0'\\
        -\Xi_0\phi_1\Lambda_0 & \Xi_0\phi_1\Lambda_0\Xi_0'+\Upsilon_0
    \end{pmatrix},
\end{align*}
which is part (ii).

\textbf{Step 7 (Case $-1<a<0$).}
Now $h_0^{-a}\to0$, so by Lemma~\ref{lemma:delta_convergence}
\begin{align*}
    \sqrt{nh_0^{1-a}}\,\Delta\delta
    = h_0^{-a}\,\sqrt{nh_0^{1+a}}\,\Delta\delta
    = h_0^{-a}O_p(1)=o_p(1);
\end{align*}
the first step is estimated at the faster rate $\sqrt{nh_0^{1+a}}$ and is
therefore asymptotically negligible at the slower rate $\sqrt{nh_0^{1-a}}$
relevant for $\hat\beta_0$. Multiplying (\ref{eq:beta0_master2}) by
$\sqrt{nh_0^{1-a}}$, and using $a<0<1$,
\begin{align*}
    \sqrt{nh_0^{1-a}}\,O_p\!\left(n^{-1/2}\right)&=O_p\!\left(h_0^{(1-a)/2}\right)=o_p(1),\\
    \sqrt{nh_0^{1-a}}\,o_p\!\left(\sqrt{h_0^{a-1}/n}\right)&=o_p(1),\\
    \sqrt{nh_0^{1-a}}\,o_p\!\left(\normm{\Delta\delta}_2\right)&=o_p\!\left(h_0^{-a}\right)=o_p(1),\\
    \sqrt{nh_0^{1-a}}\,O_p\!\left(h_0^{m+a+2}\right)
    &=O_p\!\left(\sqrt{nh_0^{a+2m+3}}\;h_0\right)=o_p(1).
\end{align*}
Hence $\sqrt{nh_0^{1-a}}(\hat\beta_0-\beta_0)=\sqrt{nh_0^{1-a}}\,\E_n\psi_i+o_p(1)$. The joint
statement follows from (\ref{eq:joint_CLT}).

\end{proof}

\noindent\textbf{Declaration of generative AI and AI-assisted technologies in the manuscript preparation process}

\noindent During the preparation of this work, the authors used Claude and ChatGPT to edit and polish the text, including the
presentation of theorems and proofs, and to write and debug code. The authors verified
all mathematical results and takes full responsibility for the content of the article.

\end{appendices}
\newpage
\bibliographystyle{chicago}
\bibliography{ref}

@article{CHAMBERLAIN1982,
title = {Multivariate regression models for panel data},
journal = {Journal of Econometrics},
volume = {18},
number = {1},
pages = {5-46},
year = {1982},
issn = {0304-4076},
doi = {https://doi.org/10.1016/0304-4076(82)90094-X},
url = {https://www.sciencedirect.com/science/article/pii/030440768290094X},
author = {Gary Chamberlain}
}

@article{CHAMBERLAIN1987,
title = {Asymptotic efficiency in estimation with conditional moment restrictions},
journal = {Journal of Econometrics},
volume = {34},
number = {2},
pages = {305-334},
year = {1987},
issn = {0304-4076},
doi = {https://doi.org/10.1016/0304-4076(87)90015-7},
url = {https://www.sciencedirect.com/science/article/pii/0304407687900157},
author = {Gary Chamberlain}
}

@article{CHAMBERLAIN1992,
 ISSN = {00129682, 14680262},
 URL = {http://www.jstor.org/stable/2951584},
 author = {Gary Chamberlain},
 journal = {Econometrica},
 number = {3},
 pages = {567--596},
 publisher = {[Wiley, Econometric Society]},
 title = {Efficiency Bounds for Semiparametric Regression},
 urldate = {2026-06-23},
 volume = {60},
 year = {1992}
}

@article{grahampowell,
author = {Graham, Bryan S. and Powell, James L.},
title = {Identification and Estimation of Average Partial Effects in “Irregular” Correlated Random Coefficient Panel Data Models},
journal = {Econometrica},
volume = {80},
number = {5},
pages = {2105-2152},
doi = {https://doi.org/10.3982/ECTA8220},
url = {https://onlinelibrary.wiley.com/doi/abs/10.3982/ECTA8220},
eprint = {https://onlinelibrary.wiley.com/doi/pdf/10.3982/ECTA8220},
year = {2012}
}

@article{PESARAN199579,
title = {Estimating long-run relationships from dynamic heterogeneous panels},
journal = {Journal of Econometrics},
volume = {68},
number = {1},
pages = {79-113},
year = {1995},
issn = {0304-4076},
doi = {https://doi.org/10.1016/0304-4076(94)01644-F},
url = {https://www.sciencedirect.com/science/article/pii/030440769401644F},
author = {M.Hashem Pesaran and Ron Smith}
}

@techreport{hausman2025linear,
  title={Fisher-Schultz Lecture: Linear Estimation of Structural and Causal Effects for Nonseparable Panel Data},
  author={Chernozhukov, Victor and Deaner, Ben and Gao, Ying and Hausman, Jerry A. and Newey, Whitney K.},
  institution={National Bureau of Economic Research},
  type={NBER Working Paper},
  number={33325},
  year={2025}
}

@misc{kwon2026estimatingtreatmenteffectsbounded,
      title={Estimating Treatment Effects Under Bounded Heterogeneity},
      author={Soonwoo Kwon and Liyang Sun},
      year={2025},
      eprint={2510.05454},
      archivePrefix={arXiv},
      primaryClass={econ.EM},
      url={https://arxiv.org/abs/2510.05454}, 
}

@article{ArellanoB2012,
  author  = {Arellano, Manuel and Bonhomme, Stephane},
   title   = {Identifying distributional characteristics in random
              coefficients panel data models},
   journal = {Review of Economic Studies},
   year    = {2012},
  volume  = {79},
   pages   = {987--1020}
 }

@article{LAAGE2024105804,
title = {A Correlated Random Coefficient panel model with time-varying endogeneity},
journal = {Journal of Econometrics},
volume = {242},
number = {2},
pages = {105804},
year = {2024},
issn = {0304-4076},
author = {Louise Laage}
}

@article{CorrelatedPanelModelLargeT,
author = {Fernández-Val, Iván and Lee, Joonhwah},
title = {Panel data models with nonadditive unobserved heterogeneity: Estimation and inference},
journal = {Quantitative Economics},
volume = {4},
number = {3},
pages = {453-481},
year = {2013}
}

@article{Swamy1970,
  author  = {Swamy, P. A. V. B.},
  title   = {Efficient Inference in a Random Coefficient Regression Model},
  journal = {Econometrica},
  year    = {1970},
  volume  = {38},
  number  = {2},
  pages   = {311--323},
  doi     = {10.2307/1913012}
}

@article{OlleyPakes1996,
  author  = {Olley, G. Steven and Pakes, Ariel},
  title   = {The Dynamics of Productivity in the Telecommunications
             Equipment Industry},
  journal = {Econometrica},
  year    = {1996},
  volume  = {64},
  number  = {6},
  pages   = {1263--1297},
  doi     = {10.2307/2171831}
}

@article{LevinsohnPetrin2003,
  author  = {Levinsohn, James and Petrin, Amil},
  title   = {Estimating Production Functions Using Inputs to Control
             for Unobservables},
  journal = {The Review of Economic Studies},
  year    = {2003},
  volume  = {70},
  number  = {2},
  pages   = {317--341},
  doi     = {10.1111/1467-937X.00246}
}

@article{ACF2015,
  author  = {Ackerberg, Daniel A. and Caves, Kevin and Frazer, Garth},
  title   = {Identification Properties of Recent Production Function
             Estimators},
  journal = {Econometrica},
  year    = {2015},
  volume  = {83},
  number  = {6},
  pages   = {2411--2451},
  doi     = {10.3982/ECTA13408}
}

@article{Raval2023,
  author  = {Raval, Devesh},
  title   = {Testing the Production Approach to Markup Estimation},
  journal = {The Review of Economic Studies},
  year    = {2023},
  volume  = {90},
  number  = {5},
  pages   = {2592--2611},
  doi     = {10.1093/restud/rdad002}
}

@article{gandhi2020identification,
  author = {Gandhi, Amit and Navarro, Salvador and Rivers, David},
  title = {On the Identification of Gross Output Production Functions},
  journal = {Journal of Political Economy},
  volume = {128},
  number = {8},
  pages = {2973--3016},
  year = {2020},
  publisher = {The University of Chicago Press},
  doi = {10.1086/707736}
}

@techreport{Greenstreet2007,
  author      = {Greenstreet, David},
  title       = {Exploiting Sequential Learning to Estimate
                 Establishment-Level Productivity Dynamics and Decision Rules},
  institution = {University of Oxford, Department of Economics},
  type        = {Economics Series Working Paper},
  number      = {345},
  year        = {2007}
}

@book{van2000asymptotic,
  title={Asymptotic statistics},
  author={Van der Vaart, Aad W},
  volume={3},
  year={2000},
  publisher={Cambridge university press}
}

@article{Sasaki_2026,
   title={SLOW MOVERS IN PANEL DATA},
   ISSN={1469-4360},
   url={http://dx.doi.org/10.1017/S0266466625100157},
   DOI={10.1017/s0266466625100157},
   journal={Econometric Theory},
   publisher={Cambridge University Press (CUP)},
   author={Sasaki, Yuya and Ura, Takuya},
   year={2026},
   month=Jan, pages={1–43} }

@incollection{MairesseGriliches1990,
  author    = {Mairesse, Jacques and Griliches, Zvi},
  title     = {Heterogeneity in Panel Data: Are There Stable Production Functions?},
  booktitle = {Essays in Honor of Edmond Malinvaud, Volume 3: Empirical Economics},
  editor    = {Champsaur, Paul and others},
  publisher = {MIT Press},
  address   = {Cambridge, MA},
  pages     = {192--231},
  year      = {1990}
}

@misc{KasaharaSchrimpfSuzuki2023,
  author        = {Kasahara, Hiroyuki and Schrimpf, Paul and Suzuki, Michio},
  title         = {Identification and Estimation of Production Function with Unobserved Heterogeneity},
  year          = {2023},
  eprint        = {2305.12067},
  archivePrefix = {arXiv},
  primaryClass  = {econ.EM},
  url           = {https://arxiv.org/abs/2305.12067}
}

@article{LiSasaki2024,
  author  = {Li, Tong and Sasaki, Yuya},
  title   = {Identification of Heterogeneous Elasticities in Gross-Output Production Functions},
  journal = {Journal of Econometrics},
  volume  = {238},
  number  = {2},
  pages   = {105637},
  year    = {2024},
  doi     = {10.1016/j.jeconom.2023.105637}
}

@article{Li2026TERC,
  author  = {Li, Ming},
  title   = {Identification and Estimation in a Time-Varying Endogenous Random Coefficient Panel Data Model},
  journal = {Journal of Econometrics},
  volume  = {256},
  pages   = {106272},
  year    = {2026},
  doi     = {10.1016/j.jeconom.2026.106272}
}

\end{document}